\documentclass[preprint,12pt]{elsarticle}
\usepackage{graphics}
\usepackage{graphicx}
\usepackage{epsfig}
\usepackage{bm}
\usepackage{subfigure}
\usepackage{multicol,balance}
\usepackage{booktabs}
\usepackage{amsthm}
\usepackage{arydshln}
\usepackage{lineno}
\usepackage{amssymb}
\usepackage{amsmath}
\usepackage{epsfig}
\usepackage{rotating}
\usepackage[hang,stable]{footmisc}
\usepackage{color}
\usepackage{arydshln}
\usepackage{latexsym}
\usepackage{amsmath}
\usepackage{amssymb}
\usepackage{tikz}

\usepackage{ifpdf}
\usepackage{CJK}

\newtheorem{theorem}{Theorem}

\newcounter{test}
\newtheorem*{problem0}{Problem P$_0$}
\newtheorem*{problemN}{Problem P$_N$}
\newtheorem*{problemP}{Problem P$_N^r$}
\newtheorem*{problemPP}{Problem Q$_N^r$}
\newtheorem*{problemPPP}{Problem Q$_{N,\alpha,\omega}^r$}
\journal{Computers and Mathematics with Applications}

\begin{document}
\begin{CJK*}{GBK}{kai}
\begin{frontmatter}



\title{ Optimal Boundary  Control for Water Hammer Suppression  in Fluid Transmission Pipelines  \footnote{This work was partially supported by the National Natural Science Foundation of China through grants F030119-61104048, 2012AA041701 and 61320106009.}}


\author{
Tehuan Chen$^a$, Chao Xu$^{a}$\footnote{Correspondence to: Chao Xu, Email: cxu@zju.edu.cn}, Zhigang Ren$^a$, Ryan Loxton$^b$}
\address{$^a$ State Key Laboratory of Industrial Control Technology and Institute of Cyber-Systems \& Control, Zhejiang University, Hangzhou, Zhejiang 310027, China.\\
$^b$ Department of Mathematics \& Statistics, Curtin University, Perth, Western Australia 6845, Australia.
}


\begin{abstract}
When  fluid flow in a pipeline is suddenly halted, a pressure surge or wave is created within the pipeline. This phenomenon, called
water hammer, can cause major damage to pipelines, including pipeline ruptures.
In this paper, we model the problem of mitigating water hammer during valve closure by an optimal boundary control problem
involving a nonlinear hyperbolic PDE system that describes  the fluid flow along the pipeline. The control variable in this system represents the valve boundary actuation implemented  at the pipeline terminus. To solve the boundary control problem, we first use {the method of lines}
to  obtain a finite-dimensional ODE model based on the original PDE system. Then, for the boundary control design, we apply the control parameterization method to obtain an approximate optimal parameter selection problem that can be solved
using nonlinear optimization techniques such as Sequential Quadratic Programming (SQP). We conclude the paper with simulation
results  demonstrating the capability of  optimal boundary control to significantly reduce flow fluctuation.

\end{abstract}


\end{frontmatter}

\section{Introduction}
%

Water hammer occurs   when  fluid moving through a pipeline  is forced to suddenly stop or change direction. This sudden change in motion, which could be due to valve closure, pump failure, or unexpected pipeline damage, causes  a pressure wave to propagate along the pipeline at high speed \cite{Anton1991, sciamarella2009water}.
The  wave speed  can be over 1000m/s, with significant pressure oscillation,
often causing loud noises and serious damage \cite{asli2010some}. In severe cases, water hammer may even cause the  pipeline to rupture,
resulting in slurry and water leakage (examples of pipeline rupture are shown in Figure \ref{water hammer damage}) \cite{erath1999modelling}.
Fluid pipeline failures due to water hammer effects are described in detail  in {\cite{schmitt2006water,jallouf2011probabilistic}}.

\begin{figure}
\centering\includegraphics[scale=0.7]{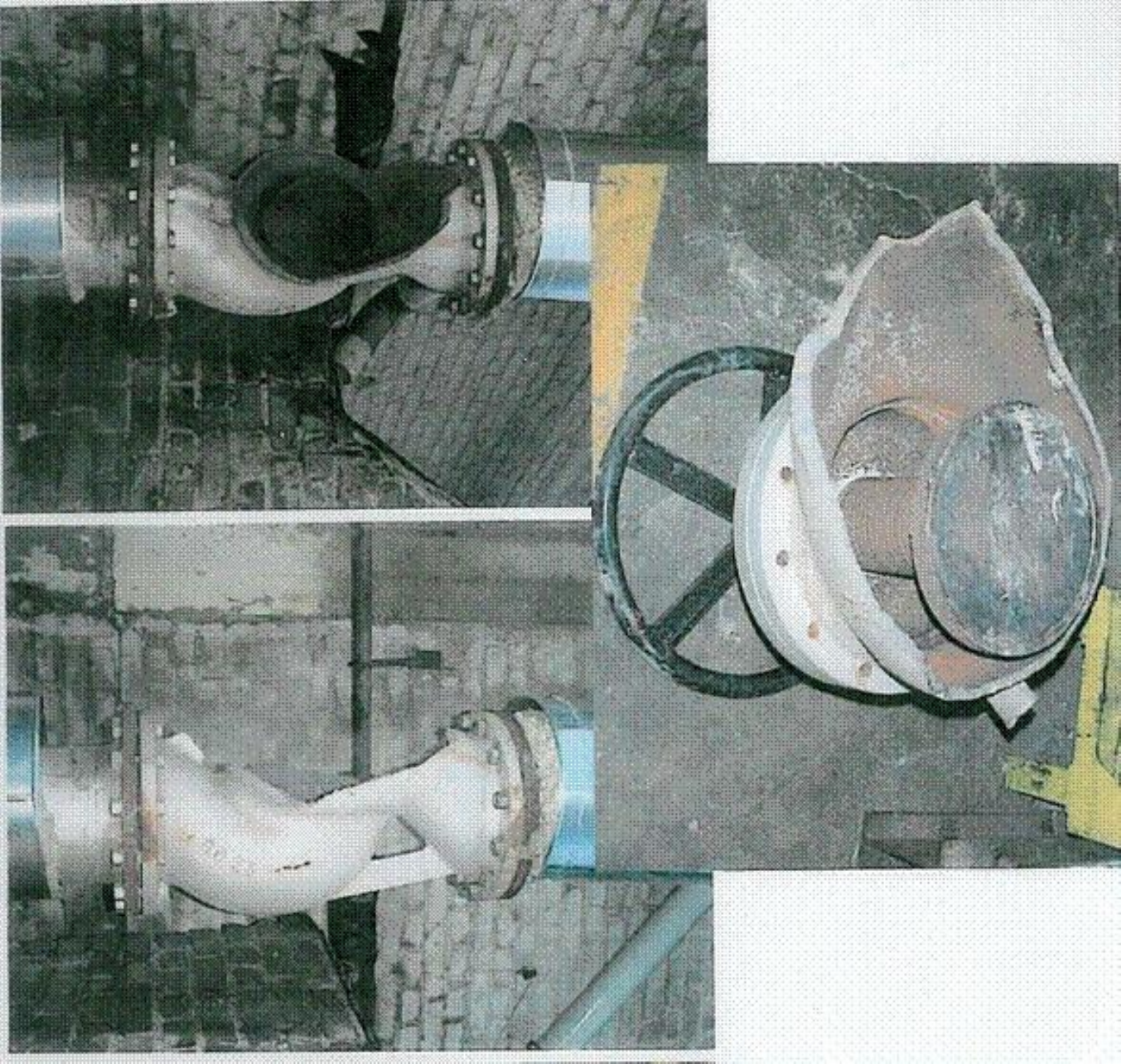}
\caption{Examples of pipeline damage caused by water hammer
(Source: http://traction.armintl.com/traction$\#$/single$\&$ proj=Docs$\&$rec=403$\&$brief=n)}
\label{water hammer damage}
\end{figure}

The mathematical equations describing water hammer   consist of  hyperbolic or parabolic partial differential equations. Numerous  methods for solving these equations, and thereby simulating water hammer,  have been developed over the past forty years. These methods  can be divided into three groups: analytical methods \cite{Cai1990}, graphical methods \cite{saikia2006simulation} and numerical methods \cite{Pierre2010}. The graphical and analytical methods are only applicable under various simplifying assumptions, and thus their value is limited in practical scenarios. In particular, the graphical and analytical methods cannot deal with the cavitation caused by negative pressure \cite{suppressing2013}.
Numerical methods for simulating
water hammer  include the fluid-structure interaction method \cite{lavooij1991fluid}, the method of characteristics {\cite{vardy1991characteristics,bouaziz2014water}}, {the heterogenous multiscale method \cite{blavzivc2014application}}, the finite volume method   \cite{zhao2004godunov}, and the wave plan method \cite{ghidaoui2005review}. In this paper, we apply {the method of lines \cite{schiesser1991numerical,schiesser2009compendium}}  to approximate the water hammer PDEs by a system of ODEs. This approach  enables the application of ODE optimal control techniques, for which there are many existing high-quality numerical algorithms, to determine optimal valve closure strategies to mitigate water hammer.

To protect a pipeline system from water hammer effects, various passive protection strategies can be employed.
These include using special materials to reinforce the pipeline and installing special devices such as relief valves, air chambers, and surge tanks   \cite{Marian2004}. However, the success of these strategies depends heavily
 on the characteristics of the  pipeline system and on the experience of the designer/operator \cite{jung2003optimum}. Moreover, although passive protection strategies can act as a guard  against water hammer, it is usually better to try and prevent water hammer from occurring in the first place. Hence, effective  control strategies for valve closure  are  required to avoid the worst effects of water hammer, such as hazardous pipeline collapse.

The  water hammer process involves nonlinearities and is non-uniform in space and time. Therefore, optimal flow control requires a forecasting model capable of predicting the non-uniform and unsteady water flow in
space and time.
Furthermore, due to flow nonlinearities, it is difficult to establish the relationship
between the control action and the corresponding response in the
hydrodynamic variables. Thus, effective valve control strategies are  essential. Cao \cite{cao2008analytic} used functional extremum theory and the   Ritz method to design optimal rules   for both velocity change and  valve closure to minimize
the peak pressure at the valve. {Axworthy} \cite{axworthy2000valve} developed  a valve closure algorithm for node-based, graph-theoretic models that can be applied within a slow transient (rigid water column)
pipeline network. {Tian} \cite{tian2008numerical} investigated  the optimum design of  parallel pump feedwater systems in nuclear power plants  to mitigate the potential damage caused by valve-induced
water hammer. Feng  \cite{wei_min2003research} proposed an optimal control method for the regulation of multiple valves, focusing  on the active  causes of water hammer.  Now,
with the rapid development of modern control theory and numerical
methodologies, 
advances in nonlinear optimization have
made the solution of nonlinear flow control problems possible. Accordingly, in this paper, we propose an effective numerical approach to determine optimal boundary controls for valve closure in fluid pipelines.

The paper is organized as follows. In Section \ref{Mathematical}, we introduce a hyperbolic PDE system to describe the fluid flow dynamics in the pipeline, after which we propose an optimal control problem for water hammer suppression during valve closure. In Section~\ref{Discrete Model_1}, we use the  {method of lines} to approximate the hyperbolic PDE system by a  non-stationary  state space ODE model. Then, in  Section \ref{method}, we use the control parameterization method, with both  piecewise-linear and  {piecewise-quadratic} basis functions,  to solve the optimal control problem by designing the boundary controller to minimize pressure fluctuation. Finally, in Section \ref{simuation}, we give numerical results to demonstrate the superiority of the optimal boundary control strategy compared with the non-optimal (but widely-used) strategy of abruptly shutting off the valve.
\section{ Problem Formulation} \label{Mathematical}
\subsection{Mathematical Model} \label{Mathematical Model}
We consider the situation shown in Figure \ref{pipeline}, where a pipeline of length $L$ is used to transport fluid from a reservoir to a terminus connected to a larger pipeline network. Let $l\in[0,L]$ denote the spatial variable along the pipeline, and let $t\in[0,T]$ denote the time variable.
By neglecting the effects of   viscosity, turbulence, and temperature variation,
the flow along the pipeline can be described by the following hyperbolic PDE system  {\cite{Wylie1993, ghidaoui2004fundamental, blavzivc2004simple}},
which consists of a momentum equation and a continuity equation:
\begin{subequations}\label{system1}
\begin{align}
&\frac{{\partial v(l,t)}}{{\partial t}} =  - \frac{1}{\rho}\frac{{\partial p(l,t)}}{{\partial l}} - \frac{{f v(l,t) \left| {v(l,t)} \right| }}{{2D}},
\label{system:1}\\
&\frac{{\partial p(l,t)}}{{\partial t}} =  - {\rho c^2 }\frac{{\partial v(l,t)}}{{\partial l}}, \label{system:2}
\end{align}
\end{subequations}
where $v$ is the
flow velocity, $p$ is the edge pressure drop, $D$ is the diameter
 of the pipeline, $c$ is the wave velocity, $f$ is the Darcy-Weisbach friction factor and $\rho$ is the flow density.\\
The boundary conditions for system (\ref{system1}) are 
\begin{figure}
\centering\includegraphics[scale=0.9]{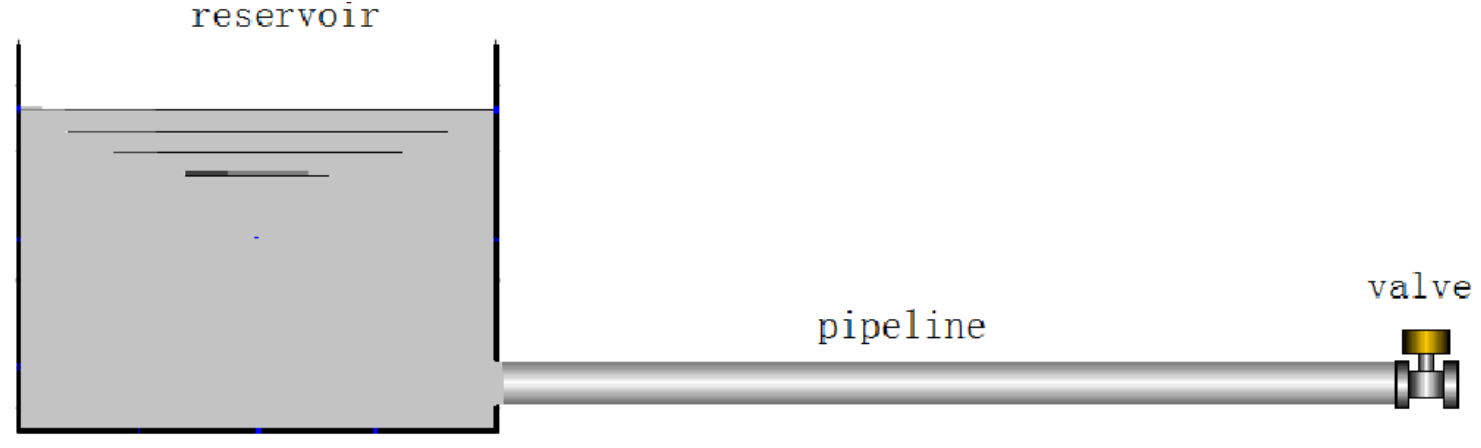}
\caption{General layout of the pipeline system described in Section \ref{Mathematical Model}}
\label{pipeline}
\end{figure}
\begin{equation}
p(0,t) =P,~~{v(L,t) = u (t)},~~t\in[0,T],\label{boundary condition}
\end{equation}
where $P$ is the   pressure generated by the reservoir (a given constant),  and $u(t)$ is a boundary control variable that models actuation from a valve situated at the pipeline terminus. Our interest is in modeling  the  fluid flow during the valve closure period, which begins at $t=0$ and ends at $t=T$. The boundary control, which must be  manipulated to implement the valve closure, is required to satisfy  the following bound constraint:
\begin{equation} \label{cons1}
0 \leq u(t)  \leq u_{\max}, ~~t\in[0,T],
\end{equation}
where $u_{\max}$ denotes the  maximum  velocity. The rate of change in the boundary control is also subject to lower and upper bounds:
\begin{equation} \label{constraints1}
-\dot{u}_{\max} \leq \dot{u}(t)  \leq \dot{u}_{\max}, ~~t\in[0,T],
\end{equation}
where $\dot{u}_{\max}$ is a given constant.  Since we require the valve to be completely closed at the terminal time,
\begin{equation}\label{control constrain1}
u(T) = 0.
\end{equation}
Any continuous function $u:[0,T] \to \mathbb{R}$ that is differentiable almost everywhere and satisfies (\ref{cons1})-(\ref{control constrain1}) is called an admissible boundary control policy.

The initial conditions for system (\ref{system1}) are
\begin{equation}
p(l,0) = \overline{p}_0 (l),~~v(l,0) = \overline{v}_0 (l),~~l\in[0,L], \label{initial conditions}
\end{equation}
where $\overline{p}_0 (l)$ and $\overline{v}_0 (l)$ are given functions describing the initial state of the pipeline.
\subsection{The Optimal Boundary Control Problem } \label{Discrete Model}
Since closing the valve suddenly could cause severe water hammer effects, the boundary control $u(t)$ must be manipulated carefully to minimize pressure fluctuation. To this end, we  consider the following objective function as proposed in references \cite{atanov1998variational,ding2006optimal}:
\begin{equation}
\begin{aligned}
J &= (p(L,T) - \hat p(L))^{2\gamma}  + \frac{1}{T}\int_0^T {(p(L,t) - \hat p(L))^{2\gamma} } dt \\
& \quad  + \frac{1}{{LT}}\int_0^T {\int_0^L {(p(l,t) - \hat p(l))^{2\gamma} dxdt} }, \label{Objective funtion}
\end{aligned}
\end{equation}
where $\gamma$ is a positive integer and $\hat p(l)$ is a given function expressing the target pressure profile along the pipeline. The objective
function (\ref{Objective funtion}) penalizes deviation between the actual  pressure in the pipeline and the target pressure profile: the first term in (\ref{Objective funtion}) penalizes pressure deviation at the valve at the terminal time, the second term penalizes  pressure deviation at the valve across the entire time horizon, and the third term penalizes global pressure deviation over the entire pipeline length and  time horizon. The reason for placing special emphasis in (\ref{Objective funtion}) at the valve location is that the valve will normally contain sensitive electrical components that must be protected. Our optimal boundary control problem is now defined as follows.
\begin{problem0}
\textit{Given the system (\ref{system1}) with boundary conditions (\ref{boundary condition}) and initial conditions (\ref{initial conditions}), choose the boundary  control $u:[0,T] \to \mathbb{R}$ to minimize the objective function (\ref{Objective funtion})  subject to the bound constraints  (\ref{cons1}) and (\ref{constraints1}) and the terminal  control constraint (\ref{control constrain1}).}
\end{problem0}
\section{Spatial Discretization} \label{Discrete Model_1}
To simplify Problem P$_0$, we will use the { method of lines} to approximate the PDE model by a state space ODE model. First, we
decompose the pipeline into equally-spaced intervals $\left[ {l_{i-1}  ,l_{i} } \right],i = 1,\dots ,N$, where $N$ is an even integer and $l_0  = 0$ and $l_{N}  = L$. Define
\begin{equation}\label{vdifferent}
\begin{aligned}
&v_i (t) = v(l_i,t), \quad i = 0,\dots,N,\nonumber
\end{aligned}
\end{equation}
and
\begin{equation}\label{pdifferent}
p_i (t) = p(l_i ,t),\quad i = 0,\dots,N.\nonumber
\end{equation}
These definitions, along with the spatial node points, are shown in Figure \ref{pipelinediffence}.
\begin{figure}
\centering\includegraphics[scale=0.6]{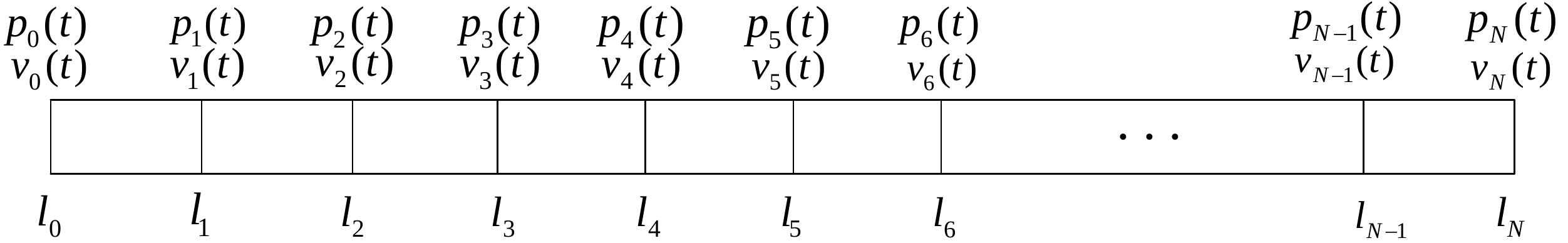}
\caption{Pipeline spatial discretization using the  {method of lines}}
\label{pipelinediffence}
\end{figure}

Based on the definitions of $v_i$ and $p_i$, we obtain the following finite difference approximations:
\begin{subequations}
\begin{align}
\hspace{-3mm}\frac{{\partial p(l_i,t)}}{{\partial l}}\hspace{-1mm} &= \hspace{-1mm}\frac{{{p_{i+1}}(t) - {p_{i }}(t)}}{\Delta l}, \quad i = 0,\ldots, N-1,\label{aproximate1}\\
\hspace{-3mm}\frac{{\partial v(l_{i},t)}}{{\partial l}} \hspace{-1mm} &= \hspace{-1mm} \frac{{{v_i}(t) - {v_{i - 1}}(t)}}{\Delta l} ,\quad i = 1,\ldots, N,\label{aproximate2}
\end{align}
\end{subequations}
where $\Delta l = L/N$. Substituting the finite difference approximations (\ref{aproximate1}) and (\ref{aproximate2}) into the PDE model (\ref{system:1}) and  (\ref{system:2}) yields
\begin{subequations}
\begin{align}
&\dot v_i (t) = \frac{1}{\rho \Delta l }(p_{i } (t) - p_{i+1} (t)) - \frac{{fv_i (t)  \left| {v_i (t)} \right| }}{{2D}}, \quad i = 0,\dots,N-1 , \label{pequation1}\\
&\dot p_i (t) = \frac{\rho c^2}{ \Delta l } (v_{i-1} (t) - v_{i } (t)), \quad i = 1,\dots,N . \label{pequation2}
\end{align}
\end{subequations}
By virtue of  the definitions of $v_i$ and $p_i$,  the boundary conditions (\ref{boundary condition}) become
\begin{equation}\label{system3}
p_0(t) = P,~ v_{N}(t)=u(t), ~~t\in[0,T].
\end{equation}
To simplify the notation, let
\begin{equation}
\begin{aligned}
\mathbf{x}(t) = \left[ {\begin{array}{*{20}c}
   {p_1 (t)} & {\cdots{}} & {p_{N } (t)} & {v_0 (t)} & {\cdots{}} & {v_{N-1 } (t)}  \\
\end{array}} \right]^T \in \mathbb{R}^{2N}, \nonumber\\
\end{aligned}
\end{equation}
\begin{equation}
\left| \mathbf{x} (t) \right|= \left[ {\begin{array}{*{20}c}
   {\left| p _1 (t)\right|} & {\cdots{}} & {\left| p _N (t)\right|} & {\left| v _0 (t)\right|} & {\cdots{}} & {\left| v _{N-1} (t)\right|}  \\
\end{array}} \right]^T \in \mathbb{R}^{2N}. \nonumber
\end{equation}
Then, by using the boundary conditions  (\ref{system3}), equations (\ref{pequation1}) and (\ref{pequation2}) can be written in compact form as
\begin{equation}
\mathbf{\dot x}(t)= A\mathbf{x}(t) +u(t)\mathbf{a} +P\mathbf{b} +B \mathbf{ x}(t)\circ\left| \mathbf{x}(t)\right|, \label{differentmodel}
\end{equation}
where $\circ$ denotes the Hadamard product, and \\
\begin{equation}
A = \left[ {\begin{array}{*{20}c}
   {0 } & {A_{12} }  \\
   {A_{21} } & {0 }  \\
\end{array}} \right] \in \mathbb{R}^{2N \times 2N}, \nonumber
\end{equation}
\begin{equation}
A_{12}  = \frac{{\rho c^2 }}{{\Delta l}} \left[ {\begin{array}{*{20}c}
   {1 } & { - 1} &  0&{\cdots{}} & 0 & 0  \\
   0 & {1 } & {-1}& {\cdots{}} & 0&0  \\
   0 &0 &1& {\cdots{}}  &  0& 0  \\
      {\vdots{}} & {\vdots{}} & {\vdots{}} & {\ddots } & {\vdots{}}& {\vdots{}}  \\
   0 & 0 & 0& {\cdots{}} & {1 } & { - 1 }  \\
   0 & 0 &0& {\cdots{}} & 0 & {1 }  \\
\end{array}} \right] \in \mathbb{R}^{N \times N},\nonumber
\end{equation}
\begin{equation}
A_{21}  = \frac{1}{{\rho \Delta l}}\left[ {\begin{array}{*{20}c}
   {-1 } & { 0} &  0&{\cdots{}} & 0 & 0  \\
   1 & -1 & 0& {\cdots{}} & 0&0  \\
   0 &1 &-1& {\cdots{}}  &  0& 0  \\
      {\vdots{}} & {\vdots{}} & {\vdots{}} & {\ddots} & {\vdots{}}& {\vdots{}}  \\
   0 & 0 & 0& {\cdots{}} & {-1 } & 0  \\
   0 & 0 &0& {\cdots{}} & 1 & {-1 }  \\
\end{array}} \right] \in \mathbb{R}^{N \times N},\nonumber
\end{equation}
\begin{equation}
\mathbf{a} =  [ {\begin{array}{*{20}c}
   0 & {\cdots{}} & {- \frac{{\rho c^2 }}{{\Delta l}}}& {0 } & {\cdots{}} & 0  \\
\end{array}}  ]^T \in \mathbb{R}^{ 2N},\nonumber
\end{equation}
\begin{equation}
\mathbf{b} = [ {\begin{array}{*{20}c}
   {0 } & {\cdots{}} & 0 &{ \frac{1}{{\rho \Delta l}} }   & {\cdots{}} & 0  \\
\end{array}}  ]^T \in \mathbb{R}^{ 2N},\nonumber
\end{equation}
\begin{equation}
B =-\frac{{f}}{{2D}} \left[ {\begin{array}{*{20}c}
   {0 } & {0}  \\
   {0 } & {I }  \\
\end{array}} \right]\in \mathbb{R}^{ 2N},\nonumber
\end{equation}
and $I$ is the $N\times N$ identity matrix. The initial conditions (\ref{initial conditions}) become
\begin{equation} \label{initial conditions2}
\begin{aligned}
\mathbf{x}(0) = \left[ {\begin{array}{*{20}c}
   \overline{p}_0 (l_1) & {\cdots{}} & \overline{p}_0 (l_N) & {\overline{v}_0 (l_0)} & {\cdots{}} & {\overline{v}_0 (l_{N-1})}  \\
\end{array}} \right]^T \in \mathbb{R}^{2N}. \\
\end{aligned}
\end{equation}
Furthermore, using  Simpson's rule \cite{gerald2003numerical}, the objective function (\ref{Objective funtion}) becomes
\begin{equation} \label{obj1}
\begin{aligned}
 J &= (x_N (T) - \hat p(L))^{2\gamma }  + \int_0^T \bigg \{ \frac{3N+1}{3NT} (x_N (t) - \hat p(L ))^{2\gamma }
 \\&\quad  + \frac{1}{3TN} ( P - \hat p(0 ))^{2\gamma }+
 \frac{4}{3TN}\sum\limits_{j = 1}^{{{N } \mathord{\left/ {\vphantom {{(N )} 2}} \right.
\kern-\nulldelimiterspace} 2}} {(x_{2j-1} (t) - \hat p(l_{2j-1} ))^{2\gamma } }
\\ &\quad  + \frac{2}{3TN}\sum\limits_{j = 1}^{{{N } \mathord{\left/
 {\vphantom {{(N) } 2}} \right.
 \kern-\nulldelimiterspace} 2}-1}{(x_{2j } (t) - \hat p(l_{2j } ))^{2\gamma } } \bigg\} dt.
\end{aligned}
\end{equation}
Note that the  term $( P - \hat p(0 ))^{2\gamma }$ in the integral is a constant. Hence, instead of minimizing (\ref{obj1}), we can equivalently  minimize the following modified objective function:
\begin{equation} \label{obj}
\begin{aligned}
 J &= (x_N (T) - \hat p(L))^{2\gamma }  + \int_0^T \bigg \{ \frac{3N+1}{3NT} (x_N (t) - \hat p(L ))^{2\gamma }
 \\&\quad  +
 \frac{4}{3TN}\sum\limits_{j = 1}^{{{N } \mathord{\left/ {\vphantom {{(N )} 2}} \right.
\kern-\nulldelimiterspace} 2}} {(x_{2j-1} (t) - \hat p(l_{2j-1} ))^{2\gamma } }  + \frac{2}{3TN}\sum\limits_{j = 1}^{{{N } \mathord{\left/
 {\vphantom {{(N) } 2}} \right.
 \kern-\nulldelimiterspace} 2}-1}{(x_{2j } (t) - \hat p(l_{2j } ))^{2\gamma } } \bigg\} dt.
\end{aligned}
\end{equation}
Our approximate problem is now stated as follows.
\begin{problemN}
\textit{ Given the system (\ref{differentmodel}) with initial condition (\ref{initial conditions2}), choose the optimal control $u:[0,T] \to \mathbb{R}$ to minimize the objective function (\ref{obj})  subject to the bound constraints (\ref{cons1}) and (\ref{constraints1})  and the terminal  control constraint (\ref{control constrain1}).}
\end{problemN}
\section{Control Parameterization } \label{method}
Problem P$_N$ is a conventional optimal control problem governed by ODEs. To solve Problem P$_N$ numerically, we will use the control parameterization
method \cite{Teo1991}, which involves approximating the control by a linear combination of basis functions, where the coefficients in the linear combination are decision variables to be optimized.  Then, by exploiting special formulae for the gradient of the objective function, the resulting approximate problem can be solved using
standard gradient-based optimization techniques \cite{ linsurvey2013}.

Control parameterization is normally applied with piecewise-constant basis functions \cite{ linsurvey2013}. However, piecewise-constant control approximation is not suitable for Problem~P$_N$ because the boundary controller in Problem~P$_N$ is required to be continuous. Thus, we instead develop two continuous approximation schemes: one with piecewise-linear basis functions, the other with piecewise-quadratic basis functions.
\subsection{Piecewise-Linear Control Parameterization}\label{Piecewise-Linear}
For piecewise-linear control parameterization \cite{lin2012optimal}, we approximate the derivative of the boundary  control    as follows:
\begin{equation}\label{signal001}
{\dot{u}(t) }\approx \sigma^k, \quad  t\in[t_{k-1}, t_k),\quad k = 1,\dots,r,
\end{equation}
where $r > 1$ is the number of approximation subintervals, $[t_{k-1}, t_k)$ is the $k$th approximation subinterval, and $\sigma^k$ is the rate of change of the control on the  $k$th subinterval. Moreover,  $t_k, ~k = 0, \dots, r,$ are fixed knot points such that
\begin{equation}
0 = t_0 < t_1 < t_2 < \dots < t_{r-1} < t_r = T.
\end{equation}
We can write equation (\ref{signal001}) as
\begin{equation}\label{approx001}
{\dot{u}(t)} \approx  \sum\limits_{k = 1}^r {\sigma ^k } \chi_{[t_{k - 1} ,t_k )}(t),
\end{equation}
where   $\chi _{[t_{k - 1} ,t_k )}(t) $ is the indicator function defined by
\begin{equation}
\chi_{[t_{k-1},t_k)}(t)=
\begin{cases}
1,&\text{if $t\in [t_{k-1},t_k)$},\\
0,&\text{otherwise}.
\end{cases}
\end{equation}
With $\dot{u}(t)$ approximated by a piecewise-constant function according to (\ref{approx001}), $u(t)$ is piecewise-linear with jumps in the derivative at $t=t_1,t_2,\ldots,t_{r-1}$. Let $x_{2N+1}(t)=u(t)$ be a new state variable. Then $x_{2N+1}(t)$ is governed by the following dynamics:
\begin{subequations}\label{new state}
\begin{align}
 &\dot x_{2N+1}(t) = \sum\limits_{k = 1}^r {{\sigma} ^k } \chi_{[t_{k - 1} ,t_k )}(t),\quad t \in [0,T],  \label{new state1}\\
 &x_{2N+1}(0)=u_0, \label{x2Ninitial}
 \end{align}
\end{subequations}
where $u_0=u_{\max}$ is the initial value of $u(t)$. In view of (\ref{cons1}),  we require the following continuous state inequality  constraint:
\begin{equation}\label{inequality constraints11}
0 \leq x_{2N+1}(t) \leq u_{\max }, \quad t \in [0,T].
\end{equation}
Clearly, since $x_{2N+1}(t)$ is piecewise-linear with break points at $t=t_1,t_2,\ldots,t_{r-1}$, this continuous state inequality  constraint is equivalent to the following constraints:
\begin{equation}\label{inequality constraints}
0 \leq x_{2N+1}(t_k ) \leq  u_{\max } ,\quad k = 1,\dots,r.\\
\end{equation}
Such constraints are known as canonical constraints in the optimal control literature \cite{linsurvey2013}.

Under the piecewise-linear control parameterization scheme (\ref{approx001}), the constraints (\ref{constraints1}) become
\begin{equation} \label{cons12}
-\dot{u}_{\max}  \le \sigma  ^k \le  \dot{u}_{\max} ,\quad k = 1,\dots,r.
\end{equation}
In addition, the dynamic system (\ref{differentmodel})  becomes
\begin{equation}\label{dynamicstrans}
\mathbf{\dot x}(t)= A\mathbf{x}(t) +x_{2N+1}(t)\mathbf{a} +P\mathbf{b} +B \mathbf{ x}(t)\circ\left| \mathbf{x}(t)\right|, \quad t\in[0, T].
\end{equation}
Furthermore,  the terminal control constraint (\ref{control constrain1}) becomes the following terminal state constraint:
\begin{equation}\label{g3cons}
x_{2N+1}(T) = 0.
\end{equation}
Our approximate problem is defined as follows.
\begin{problemP}
Given the system defined by (\ref{new state}), (\ref{dynamicstrans}), and (\ref{initial conditions2}), choose the control parameter vector
${\bm {\sigma} } = \left[ {\begin{array}{*{20}c}
   {{\sigma}} ^1   & {\cdots} & {{\sigma} ^{r} }
\end{array}} \right] \in \mathbb{R}^{r}$ to minimize the objective  function (\ref{obj}) subject to the bound constraints (\ref{cons12}) and the state constraints (\ref{inequality constraints}) and (\ref{g3cons}).
\end{problemP}
\subsection{Solving Problem P$_N^r$} \label{Psolving}
The approximate problem defined in Section \ref{Piecewise-Linear} is a  nonlinear optimization problem in which a finite number of decision variables need to be chosen to
minimize an objective function subject to a set of constraints. For this approximate problem, the
objective  function is an implicit---rather than explicit---function of the decision variables. Thus, computing the gradient of the objective function, as required to solve the approximate problem using gradient-based optimization methods such as SQP, is a non-trivial task. Nevertheless, we will now show that this gradient can be computed using the sensitivity approach described in \cite{loxton2008optimal,kaya2003computational}. Then, the SQP method can be applied  to generate search directions that lead to profitable areas of the search space \cite{luenberger2003linear}.

First, let $\mathbf{x}^r(\cdot|\mathbf{{{\bm\sigma}}})$ and ${x}^r_{2N+1}(\cdot|\mathbf{{{\bm\sigma}}})$ denote the solution of the enlarged system defined by (\ref{new state}), (\ref{dynamicstrans}), and (\ref{initial conditions2}) corresponding to the control parameter vector
${{\bm \sigma } = \left[ {\begin{array}{*{20}c}
   {{\sigma} ^1 } & {\cdots} & {{\sigma} ^{r} }  \\
\end{array}} \right]}$.  Then we have the following result.
\begin{theorem} For each $m=1,\dots,r$, the state variation of  ${x}^r_{2N+1}(\cdot|\mathbf{{{\bm\sigma}}})$ on the interval $[t_{m-1},t_m]$ is given by
\begin{equation}\label{prove1}
\frac{{\partial x_{2N + 1}^r (t|\mathbf{{{\bm\sigma}}})}}{{\partial {\sigma} ^k }} =\begin{cases}
   {t - t_{m - 1} },  &\text{if $k = m$},\\
   {t_k  - t_{k - 1} }, &\text{if $k < m$}, \\
   0,  &\text{if $k > m$}.
\end{cases}
\end{equation}
\end{theorem}
\begin{proof}
The proof is by induction. For $m=1$, it follows from (\ref{new state}) that
\begin{equation}
x_{2N + 1}^r (t|\mathbf{{{\bm\sigma}}})= u_{\max}+ {{\sigma}}^1(t-t_0)= u_{\max}+ {{\sigma}^1} t,\quad t\in[0,t_{1}].
\end{equation}
Then clearly, for all $t\in[0,t_{1}]$,
\begin{equation}
\frac{{\partial x_{2N + 1}^r (t|\mathbf{{{\bm\sigma}}})}}{{\partial {{\sigma}} ^k }}=
\begin{cases}
t ,&\text{if $k = 1$},\\
0,&\text{if $k > 1$},
\end{cases}
\end{equation}
which shows that (\ref{prove1}) is satisfied for $m=1$. Now, suppose that (\ref{prove1}) holds for $m=q$. Then for all $t\in[t_{q-1},t_{q}]$,
\begin{equation}
\frac{{\partial x_{2N + 1}^r (t|\mathbf{{{\bm\sigma}}})}}{{\partial {{\sigma}} ^k }} =\begin{cases}
   {t - t_{q - 1} },  &\text{if $k = q$},\\
   {t_k  - t_{k - 1} }, &\text{if $k < q$}, \\
   0,  &\text{if $k > q$}.
\end{cases}
\end{equation}
For $m=q+1$, equation (\ref{new state1}) implies
\begin{equation}
x_{2N + 1}^r (t|\mathbf{{{\bm\sigma}}})= x_{2N + 1}^r (t_q|\mathbf{{{\bm\sigma}}}) + {{\sigma}}^{q+1}(t-t_q),\quad t\in[t_q,t_{q+1}].
\end{equation}
Hence,  for all $t\in[t_q,t_{q+1}]$,
\begin{equation}
\frac{{\partial x_{2N + 1}^r (t|\mathbf{{{\bm\sigma}}})}}{{\partial {{\sigma}} ^k }} =\begin{cases}
   {t - t_{q} },  &\text{if $k = q+1$},\\
   \frac{{\partial x_{2N + 1}^r (t_q|\mathbf{{{\bm\sigma}}})}}{{\partial {{\sigma}} ^k }}, &\text{if $k < q+1$}, \\
   0,  &\text{if $k > q+1$}.
\end{cases}
\end{equation}
Applying the inductive hypothesis yields
\begin{equation}
\frac{{\partial x_{2N + 1}^r (t|\mathbf{{{\bm\sigma}}})}}{{\partial {{\sigma}} ^k }} =\begin{cases}
   {t - t_{q } },  &\text{if $k = q+1$},\\
   {t_k - t_{k-1 } }, &\text{if $k < q+1$}, \\
   0,  &\text{if $k > q+1$}.
\end{cases}
\end{equation}
This shows that (\ref{prove1}) holds for $m=q+1$. Thus, the result follows from {mathematical} induction.
\end{proof}
Clearly,
\begin{equation}\label{InitialSensitivity}
\frac{\partial {\mathbf{x}^r}(t|{{\bm \sigma}})}{\partial  {{\sigma}^k}}= \mathbf{0},\quad t \in [0,t_{k-1}].
\end{equation}
Moreover, for each $m=k,k+1,\ldots,r$,
\begin{equation}\label{Statechanggedform}
\begin{aligned}
{\mathbf{x}}^r(t|{{\bm \sigma}})&=\mathbf{x}^r(t_{m-1}|{{\bm \sigma}})+\int_{t_{m-1}}^t
  \Big \{  A\mathbf{x}^r(s|{{\bm \sigma}}) +\mathbf{a}x^r_{2N+1}(s|{{\bm \sigma}})
\\& \quad  +\mathbf{b}P +B \mathbf{x}^r(s|{{\bm \sigma}})\circ\left| \mathbf{x}^r(s|{{\bm \sigma}})\right|  \Big  \} ds,\quad t\in[t_{m-1},t_{m}].
\end{aligned}
\end{equation}
Differentiating (\ref{Statechanggedform}) with respect to ${{\sigma}}^k$ gives
\begin{equation}\label{diffStatechanggedform3}
\begin{aligned}
\frac{\partial {\mathbf{x}}^r(t|{{\bm \sigma}})}{\partial  {{\sigma}}^k}&= \frac{\partial \mathbf{x}^r(t_{m-1}|{{\bm \sigma}})}{\partial  {{\sigma}}^k}+\int_{t_{m-1}}^t \left\{A\frac{\partial \mathbf{x}^r(s|{{\bm \sigma}})}{\partial {{\sigma}}^k}+\mathbf{a}\frac{\partial x^r_{2N+1}(s|{{\bm \sigma}})}{\partial {{\sigma}}^k} \right.
\\ \quad & \left. + 2 B\left| \mathbf{x}^r(s|{{\bm \sigma}}\right)|\circ\frac{\partial \mathbf{x}^r(s|{{\bm \sigma}})}{\partial {{\sigma}}^k}\right \} ds,\quad t\in[t_{m-1},t_{m}),\quad m=k,k+1,\ldots,r.
\end{aligned}
\end{equation}
Thus, differentiating (\ref{diffStatechanggedform3}) with respect to time $t$, we obtain
\begin{equation}\label{sensitityform001}
\begin{aligned}
&\frac{d}{dt}\left\{\frac{\partial {\mathbf{x}^r}(t|{{\bm \sigma}})}{\partial  {{\sigma}}^k}\right\}= A\frac{\partial \mathbf{x}^r(t|{{\bm \sigma}})}{\partial {{\sigma}}^k}+\mathbf{a}\frac{\partial x^r_{2N+1}(t|{{\bm \sigma}})}{\partial {{\sigma}}^k}\\
&\quad+ 2B\left| \mathbf{x}^r(t|{{\bm \sigma}}\right)|\circ\frac{\partial \mathbf{x}^r(t|{{\bm \sigma}})}{\partial {{\sigma}}^k},\quad t\in[t_{m-1},t_{m}),\quad m=k,k+1,\ldots,r.
\end{aligned}
\end{equation}
Based on (\ref{InitialSensitivity}) and (\ref{sensitityform001}), we have the following result.
\begin{theorem} The state variation of $\mathbf{x}^r(\cdot|\mathbf{{{\bm\sigma}}})$ with respect to ${{\sigma}}^k$ is the solution $\Gamma^k(\cdot|\mathbf{{{\bm\sigma}}})$ of the following sensitivity system:
\begin{equation}\label{sensitityform01}
\begin{aligned}
\dot \Gamma^{k}(t)= A\Gamma^{k}(t)+\mathbf{a}\frac{{\partial x_{2N + 1}^r (t|\mathbf{{{\bm\sigma}}})}}{{\partial {{\sigma}} ^k }}+ 2 B \left| \mathbf{x}^r(t|{{\bm \sigma}}\right)|\circ\Gamma^{k}(t),\\ \quad t\in[t_{m-1},t_{m}),\quad m=k,k+1,\ldots,r,
\end{aligned}
\end{equation}
where $ \Gamma^{k}(t)=0$, $t\in[0,t_{k-1}) $ and $\frac{{\partial x_{2N + 1}^r (t|\mathbf{{{\bm\sigma}}})}}{{\partial {{\sigma}} ^k }}$ is given by the formula in  Theorem~1.
\end{theorem}
Clearly, the gradients of constraints (\ref{inequality constraints}) and (\ref{g3cons}) can be computed using Theorem~1.
For the objective function, the gradient can be obtained by differentiating (\ref{obj}) using the chain rule: 
\begin{equation}  \nonumber
\begin{aligned}
\frac{\partial J({{\bm \sigma}} )} {{\partial {{ \sigma}}^k }} &=2\gamma(x_N (T) - \hat p(L))^{2\gamma-1}\Gamma^{k}_{N}( T\left| \mathbf{ {{\bm \sigma}}}  \right.)
\\&\quad+  \int_0^{T} \bigg  \{ \frac{2\gamma(3N+1)}{3NT} {(x_N (t) - \hat p(L ))^{2\gamma-1 }}\Gamma^{k}_N(t\left| \mathbf{ {{\bm \sigma}}}  \right.)
 \\&\quad  + \frac{8\gamma}{3NT}\sum\limits_{j = 1}^{{{N } \mathord{\left/ {\vphantom {{(N )} 2}} \right.
\kern-\nulldelimiterspace} 2}} {(x_{2j-1} (t) - \hat p(l_{2j-1} ))^{2\gamma-1 } }\Gamma^{k}_{2j-1}(t\left| \mathbf{ {{\bm \sigma}}}  \right.)
  \\& \quad  + \frac{4\gamma}{3NT}\sum\limits_{j = 1}^{{{N } \mathord{\left/
 {\vphantom {{(N) } 2}} \right.
 \kern-\nulldelimiterspace} 2}-1} {(x_{2j } (t) - \hat p(l_{2j } ))^{2\gamma-1 }}\Gamma^{k}_{2j}(t\left| \mathbf{{{\bm \sigma}}}  \right.)  \bigg   \} dt.
 \end{aligned}
\end{equation}
By incorporating these gradient formulae  with a nonlinear programming algorithm such as SQP, Problem P$^r_N$ can be solved efficiently. The gradient-based optimization framework is illustrated in Figure \ref{SQLlc}. Convergence results showing that the solution of Problem~P$^r_N$ converges to the solution of Problem P$_N$ are derived in {\cite{Loxton20092250,loxton2012control,liu2014}}.
\begin {figure}
\begin{center}

\begin{tikzpicture}[scale=0.75]

\node at (0,0) [draw,red!40, rounded corners,fill=red!40, text width=8cm,     text centered, minimum height=0.8cm] {};

\node at (0,0) {{\footnotesize Choose initial guesses for the control parameters}};

\node at (0,-1.75) [draw,red!40, rounded corners,fill=red!40, text width=6cm,     text centered, minimum height=0.8cm] {};

\node at (0,-1.75) {{\footnotesize Solve the state and sensitivity systems}};

\node at (0,-3.5) [draw,red!40, rounded corners,fill=red!40, text width=7.5cm,     text centered, minimum height=0.8cm] {};

\node at (0,-3.5) {{\footnotesize Compute the objective and constraint gradients}};

\node at (0,-5.25) [draw,red!40, rounded corners,fill=red!40, text width=9.2cm,     text centered, minimum height=0.8cm] {};

\node at (0,-5.25) {{\footnotesize Use the gradient information to perform an optimality test}};

\node at (0,-7) [draw,green!40, rounded corners,fill=green!40, text width=1.4cm,     text centered, minimum height=0.8cm] {};

\node at (0,-7) {{\footnotesize Optimal?}};

\node at (0,-8.75) [draw,red!40, rounded corners,fill=red!40, text width=9.2cm,     text centered, minimum height=0.8cm] {};

\node at (0,-8.75) {{\footnotesize Use the gradient information to calculate a search direction}};

\node at (3.55,-7) [draw,green!40, rounded corners,fill=green!40, text width=1.4cm,     text centered, minimum height=0.8cm] {};

\node at (3.55,-7) {{\footnotesize Finish}};

\node[blue] at (0.4,-7.85) {{\footnotesize\emph{No}}};

\node[blue] at (1.75,-6.75) {{\footnotesize\emph{Yes}}};

\node at (0,-10.5) [draw,red!40, rounded corners,fill=red!40, text width=7.5cm,     text centered, minimum height=0.8cm] {};

\node at (0,-10.5) {{\footnotesize Perform a line search along the search direction}};

\node at (0,-12.25) [draw,red!40, rounded corners,fill=red!40, text width=6cm,     text centered, minimum height=0.8cm] {};

\node at (0,-12.25) {{\footnotesize Update the control parameter values}};

\draw[->,black,thick] (0,-0.54)--(0,-1.21);

\draw[->,black,thick] (0,-2.29)--(0,-2.96);

\draw[->,black,thick] (0,-4.04)--(0,-4.71);

\draw[->,black,thick] (0,-5.79)--(0,-6.46);

\draw[->,black,thick] (0,-7.54)--(0,-8.21);

\draw[->,black,thick] (1.12,-7)--(2.4,-7);

\draw[->,black,thick] (0,-9.29)--(0,-9.96);

\draw[->,black,thick] (0,-11.04)--(0,-11.71);

\draw[-,black,thick] (-4.185,-12.25)--(-8.25,-12.25);

\draw[-,black,thick] (-8.25,-12.25)--(-8.25,-1.75);

\draw[->,black,thick] (-8.25,-1.75)--(-4.185,-1.75);

\draw[-,white] (5.135,-12.25)--(7,-12.25);

\draw[-,white] (7,-12.25)--(7,-1.75);

\draw[->,white] (7,-1.75)--(5.135,-1.75);

\draw[dashed,thick,blue] (-7,-2.625) rectangle (7,-11.375);

\node[blue,rotate=90] at (7.5,-7) {\footnotesize Optimization algorithm};

\node[blue,rotate=90] at (8,-7) {\footnotesize (e.g., SQP)};

\end{tikzpicture}
\caption{Gradient-based optimization framework for solving  Problem P$^r_N$}
\label{SQLlc}
\end{center}
\end{figure}
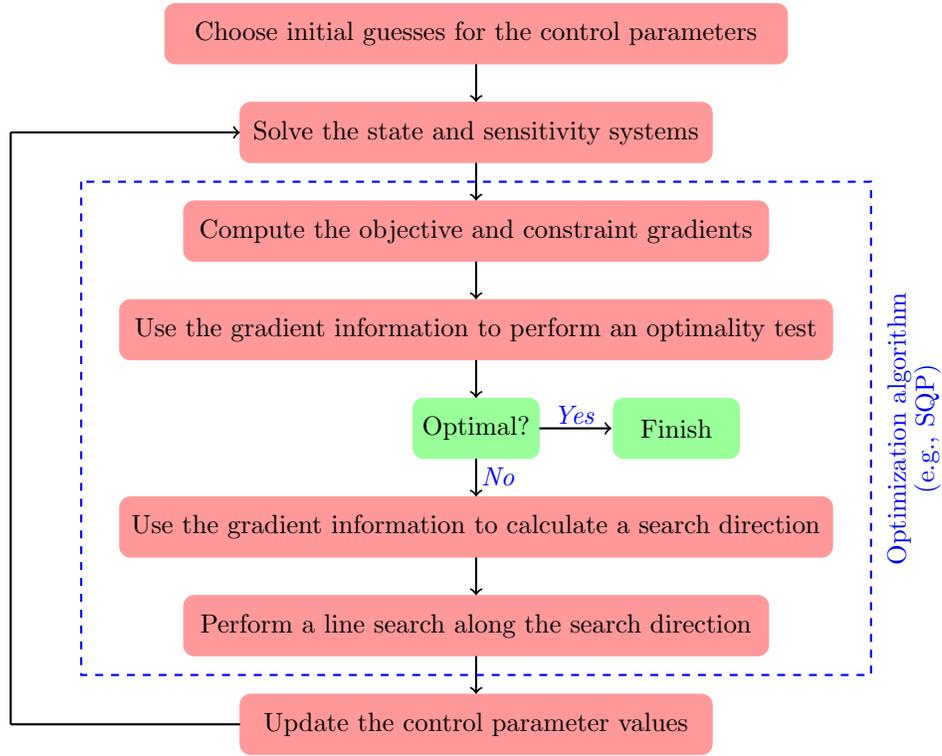
\subsection{Piecewise-Quadratic Control Parameterization}\label{Piecewise-Quardratic}
For piecewise-quadratic control parameterization, we approximate the second derivative of the control instead of the first derivative:
\begin{equation}
\ddot{{u}}(t) \approx \sum\limits_{k = 1}^r {\hat{{\sigma}} ^k } \chi_{[t_{k - 1} ,t_k )}(t),\quad t \in [0,T].
\end{equation}
Then, we introduce two new state variables $x_{2N+1}(t)$ and $x_{2N+2}(t)$  governed by the following dynamics:
\begin{subequations}\label{new state quardic}
\begin{align}
 \dot x_{2N+1}(t) &=  x_{2N+2}(t), \quad t \in [0,T], \\
  x_{2N+1}(0)&=u_{0},\\
 \dot x_{2N+2}(t) &= \sum\limits_{k = 1}^r {{\hat{\sigma}} ^k } \chi_{[t_{k - 1} ,t_k )}(t),\quad t \in [0,T],  \label{new state001}\\
 x_{2N+2}(0)&= \dot{u}_0,
 \end{align}
\end{subequations}
where $u_{0}$ and $\dot{u}_0$ are  given constants. Here, $x_{2N+2}(t)$ represents $\dot{u}(t)$ (a piecewise-linear function) and $x_{2N+1}(t)$ represents $u(t)$ (a piecewise-quadratic function). Thus, $u_0=u_{\max}$  (the valve is initially fully open) and $\dot{u}_0$ is the initial value of $\dot{u}(t)$. In view of (\ref{cons1}) and (\ref{constraints1}),  we have the following continuous state inequality  constraints:
\begin{equation}\label{constraints0011}
0 \leq x_{2N+1}(t) \leq u_{\max }, \quad t \in [0,T],
\end{equation}
and
\begin{equation}\label{inequality constraints0011}
-\dot{u}_{\max } \leq x_{2N+2}(t) \leq \dot{u}_{\max }, \quad t \in [0,T].
\end{equation}
Since $x_{2N+2}(t)$ is piecewise-linear, the continuous state inequality  constraint (\ref{inequality constraints0011}) is equivalent to:
\begin{equation}\label{inequality constraints2000}
-\dot{u}_{\max } \leq x_{2N+2}(t_k) \leq \dot{u}_{\max },\quad k = 1,\dots,r.\\
\end{equation}
After applying the piecewise-quadratic control parameterization scheme, the dynamic system (\ref{differentmodel})  becomes
\begin{equation}\label{dynamicstrans-quardratic}
\mathbf{\dot x}(t)= A\mathbf{x}(t) +x_{2N+1}(t)\mathbf{a} +P\mathbf{b} +B \mathbf{ x}(t)\circ\left| \mathbf{x}(t)\right|, \quad t\in[0,T].
\end{equation}
Furthermore,  the terminal control constraint (\ref{control constrain1}) becomes the following terminal state constraint:
\begin{equation}\label{g3cons-quardratic}
x_{2N+1}(T) = 0.
\end{equation}
Our approximate problem is defined as follows.
\begin{problemPP}
Given the system defined by (\ref{new state quardic}), (\ref{dynamicstrans-quardratic}), and (\ref{initial conditions2}), choose the control parameter vector
${\bm{\hat {\sigma}} } = \left[ {\begin{array}{*{20}c}
   {{\hat{\sigma}}} ^1   & {\cdots} & {{\hat{\sigma}} ^{r} }
\end{array}} \right] \in \mathbb{R}^{r}$ to minimize the objective  function (\ref{obj}) subject to the state constraints  (\ref{constraints0011}), (\ref{inequality constraints2000}) and (\ref{g3cons-quardratic}).
\end{problemPP}
\subsection{Solving Problem Q$_N^r$} \label{Qsolving}
Like Problem P$^r_N$, Problem Q$_N^r$ is a nonlinear optimization problem. The only significant difference is that Problem Q$_N^r$ contains a continuous inequality state constraint (\ref{constraints0011}) that cannot be converted into a finite number of conventional constraints. To address this difficulty, we first note that (\ref{constraints0011}) is equivalent to the following non-smooth integral constraints:
\begin{equation}\label{continuous in001}
\int_0^T  \max  \{-x_{2N+1}(t), 0\} dt=0, \quad  \int_0^T  \max \{x_{2N+1}(t)-{u}_{\max },0\} dt= 0.
\end{equation}
Since the $\max\{\cdot,0\}$ function is non-smooth, we use the following smooth approximation scheme defined in \cite{liu2014}:
\begin{equation}
\max \{y,0\}\approx \phi_{\alpha}(y)=\frac{1}{2}\sqrt{y^2+4\alpha^2}+\frac{1}{2}y,
\end{equation}
where $\alpha>0$ is a smoothing parameter. Note that $\phi_{\alpha}(y)\geq 0$ for all $y$.
Based on this  approximation scheme, we append constraints (\ref{continuous in001}) to the objective (\ref{obj}) to obtain the following penalty function:
\begin{equation}\label{obj-change}
\begin{aligned}
G_{\alpha,\omega}(\bm{\hat{\sigma}})&=J(\bm {\hat{\sigma}})+\omega\bigg\{\int_0^T \phi_{\alpha}(-x_{2N+1}(t)) dt
 +\int_0^T  \phi_{\alpha}(x_{2N+1}(t)-u_{\max}) dt\bigg\},\\
\end{aligned}
\end{equation}
where $\omega>0$ is a penalty parameter. We now define an approximation of Problem Q$_N^r$ as follows.
\begin{problemPPP}
Given the system defined by (\ref{new state quardic}), (\ref{dynamicstrans-quardratic}), and (\ref{initial conditions2}), choose the control parameter vector
${\bm{\hat {\sigma}} } = \left[ {\begin{array}{*{20}c}
   {{\hat{\sigma}}} ^1   & {\cdots} & {{\hat{\sigma}} ^{r} }
\end{array}} \right] \in \mathbb{R}^{r}$ to minimize the penalty  function (\ref{obj-change}) subject to the state constraints   (\ref{inequality constraints2000}) and (\ref{g3cons-quardratic}).
\end{problemPPP}
Note that when $\alpha$ is small, $\phi_{\alpha}(y)$ is a good approximation of $\max \{y,0\}$,  and thus Problem Q$_{N,\alpha,\omega}^r$ is a good approximation of Problem Q$_N^r$. Formal convergence results are given in \cite{liu2014}.

Let $\mathbf{x}^r(\cdot|\mathbf{{{\hat{\bm\sigma}}}})$, ${x}^r_{2N+1}(\cdot|\mathbf{{{\hat{\bm\sigma}}}})$, and ${x}^r_{2N+2}(\cdot|\mathbf{{{\hat{\bm\sigma}}}})$ denote the solution of the enlarged system defined by (\ref{new state quardic}), (\ref{dynamicstrans-quardratic}), and (\ref{initial conditions2}) corresponding to the control parameter vector
${\bm{\hat{ \sigma} } = \left[ {\begin{array}{*{20}c}
   {{\hat{\sigma}} ^1 } & {\cdots} & {{\hat{\sigma}} ^{r} }  \\
\end{array}} \right]}$.
Based on Theorem~1, for $t\in [t_{m-1},t_m]$,
\begin{equation}\label{BasedTheorem1}
\frac{{\partial x_{2N + 2}^r (t|\mathbf{{{\hat{\bm\sigma}}}})}}{{\partial \hat{{\sigma}} ^k }} =\begin{cases}
   {t - t_{m - 1} },  &\text{if $k = m$},\\
   {t_k  - t_{k - 1} }, &\text{if $k < m$}, \\
   0,  &\text{if $k > m$}.
\end{cases}
\end{equation}
The derivative of ${x_{2N + 1}^r (t|\mathbf{{\bm{\hat{\sigma}}}})}$ is given in the following theorem.
\begin{theorem} For each $m=1,\dots,r$, the state variation of  ${x}^r_{2N+1}(\cdot|\mathbf{{{\hat{\bm\sigma}}}})$ on the interval $[t_{m-1},t_m]$ is given by
\begin{equation}\label{prove03}
\frac{{\partial x_{2N + 1}^r (t|\mathbf{{\bm{\hat{\sigma}}}})}}{{\partial \hat{{\sigma}} ^k }} =\begin{cases}
   {\frac{1}{2}t^2-t_{m-1}t+\frac{1}{2}t_{m-1}^2 },  &\text{if $k = m$},\\
   {(t_k  - t_{k - 1})t+ \frac{1}{2} t_{k-1}^2-\frac{1}{2} t_{k}^2}, &\text{if $k < m$}, \\
   0,  &\text{if $k > m$}.
\end{cases}
\end{equation}
\end{theorem}
\begin{proof}
The proof is by induction on $m$.
For $m=1$,
\begin{equation}
\begin{aligned}
{x_{2N + 1}^r (t|\mathbf{{\bm{\hat{\sigma}}}})}&=u_0+\int ^t_{0} { x_{2N + 2}^r (s|{{\bm{\hat{\sigma}}}})} ds, \quad t \in [0,t_1].
\end{aligned}
\end{equation}
Thus, using (\ref{BasedTheorem1}), for all $t \in [0,t_1]$,
\begin{equation}
\begin{aligned}
\frac{{\partial x_{2N + 1}^r (t|\mathbf{{{\hat{\bm\sigma}}}})}}{{\partial \hat{{\sigma}} ^k }}&=\int ^t_{0} \frac{{\partial x_{2N + 2}^r (s|\mathbf{{{\hat{\bm\sigma}}}})}}{{\partial \hat{{\sigma}} ^k }} ds=\begin{cases}
\frac {1}{2}t ^2,&\text{if $k = 1$},\\
0,&\text{if $k > 1$}.
\end{cases}
\end{aligned}
\end{equation}
This shows that (\ref{prove03}) is satisfied for $m=1$. Now, suppose that (\ref{prove03})  holds for $m=q$. Then for all $t\in[t_{q-1},t_{q}]$,
\begin{equation}\label{Then for all}
\frac{{\partial x_{2N + 1}^r (t|\mathbf{{\bm{\hat{\sigma}}}})}}{{\partial \hat{{\sigma}} ^k }} =\begin{cases}
   {\frac{1}{2}t^2-t_{q-1}t+\frac{1}{2}t_{q-1}^2 },  &\text{if $k = q$},\\
   {(t_k  - t_{k - 1})t+ \frac{1}{2} t_{k-1}^2-\frac{1}{2} t_{k}^2}, &\text{if $k < q$}, \\
   0,  &\text{if $k > q$}.
\end{cases}
\end{equation}
For $t\in[t_{q},t_{q+1}]$,
\begin{equation}\label{Different}
\begin{aligned}
{ x_{2N + 1}^r (t|\mathbf{{{\hat{\bm\sigma}}}})}&={{ x_{2N + 1}^r (t_{q}|\mathbf{{{\hat{\bm\sigma}}}})}}+\int ^t_{t_{q}} {{ x_{2N + 2}^r (s|\mathbf{{{\hat{\bm\sigma}}}})}} ds.
\end{aligned}
\end{equation}
Differentiating (\ref{Different}) with respect to ${\hat{{\sigma}}}^k$ gives
\begin{equation}\label{Differentwan}
\begin{aligned}
\frac{{\partial x_{2N + 1}^r (t|\mathbf{{{\hat{\bm\sigma}}}})}}{{\partial \hat{{\sigma}} ^k }}&=\frac{{\partial x_{2N + 1}^r (t_{q}|\mathbf{{{\hat{\bm\sigma}}}})}}{{\partial \hat{{\sigma}} ^k }}+\int ^t_{t_{q}} \frac{{\partial x_{2N + 2}^r (s|\mathbf{{{\hat{\bm\sigma}}}})}}{{\partial \hat{{\sigma}} ^k }} ds.
\end{aligned}
\end{equation}
Thus, if $k>q+1$, then clearly
\begin{equation}\label{k>q+1}
\frac{{\partial x_{2N + 1}^r (t|\mathbf{{{\hat{\bm\sigma}}}})}}{{\partial \hat{{\sigma}} ^k }}=0.
\end{equation}
If $k=q+1$, then by using  (\ref{BasedTheorem1}) and  (\ref{Then for all}) to  simplify (\ref{Differentwan}), we obtain
\begin{equation}\label{k=q+1}
\begin{aligned}
\frac{{\partial x_{2N + 1}^r (t|\mathbf{{{\hat{\bm\sigma}}}})}}{{\partial \hat{{\sigma}} ^k }}&=\int ^t_{t_{q}} (s-t_q) ds=\frac {1}{2}t^2-t_qt+\frac {1}{2}t^2_q.
\end{aligned}
\end{equation}
Finally, if $k<q+1$, then the inductive hypothesis (\ref{Then for all})  implies
\begin{equation*}
\begin{aligned}
\frac{{\partial x_{2N + 1}^r (t_q|\mathbf{{{\hat{\bm\sigma}}}})}}{{\partial \hat{{\sigma}} ^k }}&=\begin{cases}
\frac {1}{2}t_q^2-t_qt_{q-1}+\frac {1}{2}t^2_{q-1},&\text{if $k = q$},\\
(t_k-t_{k-1})t_q+\frac {1}{2}t^2_{k-1}-\frac{1}{2}t^2_k, &\text{if $k<q$},
\end{cases}
\\ \quad &=(t_k-t_{k-1})t_q+\frac {1}{2}t^2_{k-1}-\frac{1}{2}t^2_k.
\end{aligned}
\end{equation*}
Thus, (\ref{Differentwan}) becomes
\begin{equation}\label{k<q+1}
\begin{aligned}
\frac{{\partial x_{2N + 1}^r (t|\mathbf{{{\hat{\bm\sigma}}}})}}{{\partial \hat{{\sigma}} ^k }}&=  {(t_k  - t_{k - 1})t_q+ \frac{1}{2} t_{k-1}^2-\frac{1}{2} t_{k}^2}+\int ^t_{t_{q}} (t_k-t_{k-1})ds
\\&=(t_k-t_{k-1})t+ \frac{1}{2}t_{k-1}^2-\frac{1}{2}t_k^2 .
\end{aligned}
\end{equation}
 Equations (\ref{k>q+1}), (\ref{k=q+1}) and (\ref{k<q+1}) show that (\ref{prove03}) holds for $m=q+1$. Thus, the result follows from {mathematical} induction.
\end{proof}
The state variation of $\mathbf{x}^r(\cdot|\mathbf{{{\hat{\bm\sigma}}}})$ in Problem Q$^r_{N,\alpha,\omega}$ can be computed in the same manner as for Problem~P$_N^r$. This leads to the following theorem (see Theorem 2).
\begin{theorem} The state variation of $\mathbf{x}^r(\cdot|\mathbf{{{\hat{\bm\sigma}}}})$ with respect to ${{\hat{\sigma}}}^k$ is the solution $\Psi^k(\cdot|\mathbf{{{\hat{\bm\sigma}}}})$ of the following sensitivity system:
\begin{equation}\label{sensitityform02}
\begin{aligned}
\dot \Psi^{k}(t)= A\Psi^{k}(t)+\mathbf{a}\frac{{\partial x_{2N + 1}^r (t|\mathbf{{{\hat{\bm\sigma}}}})}}{{\partial {{\hat{\sigma}}} ^k }}+ 2 B \left| \mathbf{x}^r(t|{{\hat{\bm \sigma}}}\right)|\circ\Psi^{k}(t),\\ \quad t\in[t_{m-1},t_{m}),\quad m=k,k+1,\ldots,r,
\end{aligned}
\end{equation}
where $ \Psi^{k}(t)=0$, $t\in[0,t_{k-1}) $ and $\frac{{\partial x_{2N + 1}^r (t|\mathbf{{{\hat{\bm\sigma}}}})}}{{\partial {{\hat{\sigma}}} ^k }}$ is given by the formula in Theorem~3.
\end{theorem}
Clearly, the gradients of constraints (\ref{inequality constraints2000}) and (\ref{g3cons-quardratic}) can be computed using equations (\ref{BasedTheorem1}) and (\ref{prove03}).
For the penalty function  (\ref{obj-change}), the gradient can be obtained  using the chain rule of differentiation: 
\begin{equation}  \nonumber
\begin{aligned}
\frac{\partial G_{\alpha,\omega}({\bm{\hat{ \sigma}}} )} {{\partial {{ \hat{\sigma}}}^k }} &=2\gamma(x_N (T) - \hat p(L))^{2\gamma-1}\Psi^{k}_{N}( T\left| \mathbf{ {{\hat{\bm \sigma}}}}  \right.) +  \int_0^{T} \bigg  \{ \frac{2\gamma(3N+1)}{3NT} {(x_N (t) - \hat p(L ))^{2\gamma-1 }}\Psi^{k}_N(t\left| \mathbf{ {{\hat{\bm \sigma}}}}  \right.)
 \\& \quad+ \frac{8\gamma}{3NT}\sum\limits_{j = 1}^{{{N } \mathord{\left/ {\vphantom {{(N )} 2}} \right.
\kern-\nulldelimiterspace} 2}} {(x_{2j-1} (t) - \hat p(l_{2j-1} ))^{2\gamma-1 } }\Psi^{k}_{2j-1}(t\left| \mathbf{ {{\hat{\bm \sigma}}}}  \right.)
  \\&\quad + \frac{4\gamma}{3NT}\sum\limits_{j = 1}^{{{N } \mathord{\left/
 {\vphantom {{(N) } 2}} \right.
 \kern-\nulldelimiterspace} 2}-1} {(x_{2j } (t) - \hat p(l_{2j } ))^{2\gamma-1 }}\Psi^{k}_{2j}(t\left| \mathbf{ {{\hat{\bm \sigma}}}}  \right.)  \bigg   \} dt
 \\ &\quad+{\omega}\int_0^T \bigg\{ \frac{d \phi_{\alpha}(x_{2N+1}(t)-u_{\max})}{d y}\frac{\partial x_{2N+1}(t|\mathbf{{{\hat{\bm\sigma}}}})}{\partial \hat{\sigma}^k}-\frac{d\phi_{\alpha}(-x_{2N+1}(t))}{d y}\frac{\partial x_{2N+1}(t|\mathbf{{{\hat{\bm\sigma}}}})}{\partial \hat{\sigma}^k} \bigg\}dt,
 \end{aligned}
\end{equation}
where $\frac{\partial x_{2N+1}(t|\mathbf{{{\hat{\bm\sigma}}}})}{\partial \hat{\sigma}^k}$ is given by the formula in Theorem 3.
By incorporating these gradient formulae  into a nonlinear programming algorithm such as SQP, Problem Q$^r_{N,\alpha,\omega}$ can be solved efficiently.   When $\alpha$ is small and $\omega$ is large, the solution of Problem Q$^r_{N,\alpha,\omega}$ is a good approximation of the solution of Problem Q$^r_{N}$. See the convergence results in \cite{liu2014} for more details.

\section{Numerical Simulations}\label{simuation}

For the numerical  simulations, we consider a stainless steel pipeline of length   $L=20$~meters and diameter  $D=100$~millimeters. The flow density is taken as $\rho= 1000$~kg/m$^3$. Since the Darcy-Weisbach friction factor $f$ for stainless steel pipelines is normally contained in the range $[0.02,0.04]$ (see reference \cite{Zhaoxin2009}), we choose $f=0.03$. Moreover, as in \cite{Yan1986}, we choose $c=1200$~m/s for the wave speed. The reservoir pressure is set at $P=2\times10^5$~Pa, which corresponds to the pressure exerted by a fluid tower  approximately 20 meters high. We assume that the pipeline fluid flow is initially in the steady state with constant velocity  $\bar{v}_0(l)=2~ \textup{m/s}$. 
It then follows from (\ref{system:1}) that
\begin{equation*}
0 =  - \frac{1}{\rho}\frac{{\partial \bar{p}_0(l)}}{{\partial l}} - \frac{{2f  }}{{D}},
 \end{equation*}
and thus
\begin{equation*}
\frac{{\partial \bar{p}_0(l)}}{{\partial l}}=- \frac{{2{\rho}f  }}{{D}}.
\end{equation*}
Integrating for $\bar{p}_0(l)$ yields
\begin{equation*}\label{initial pressure}
\bar{p}_0(l)=P-\frac{2\rho f }{D}l.
\end{equation*}
We choose  $\gamma=2$ in the objective function (\ref{obj}). In our numerical experience, larger values of $\gamma$ have little effect on the results---this is consistent with the observations in reference \cite{atanov1998variational}, which advocates $\gamma=2$ as the best choice. For the control bounds, we set $u_{\max}=2$ and $\dot{u}_{\max}=10$, and for the terminal time, we set $T=10$ seconds. Moreover, we define $\hat{p}(l)=P=2\times 10^5$~Pa as the target pressure profile, since when the valve is completely closed the pressure will be constant across the pipeline (and equal to the reservoir pressure) in steady state.


Our numerical simulation study was carried out within the MATLAB programming environment (version  R2010b) running on a personal computer with
the following configuration: Intel Core i5-2320 3.00GHz CPU, 4.00GB RAM, 64-bit Windows 7 Operating System. Our MATLAB code implements the gradient-based optimization procedure in Figure \ref{SQLlc} by
combining FMINCON with the sensitivity method for gradient computation.

\subsection{Piecewise-Linear Control Parameterization}\label{Piecewise-Linear_1}

Using the piecewise-linear control parameterization method with $r = 10$ subintervals, we solved Problem P$^r_N$ for $N=16, 18, 20, 22, 24$. Our MATLAB program uses the in-built differential equation solver ODE23 to solve the state system (\ref{dynamicstrans}) and the sensitivity systems (\ref{sensitityform01}) and (\ref{sensitityform02}).

The optimal objective function values are given in Table \ref{tab002}. Moreover, the optimal control parameters for $N=24$ are given in Table~\ref{tab2}. According to equation~(\ref{new state1}), the optimal values in Table~\ref{tab2} are the slopes of the optimal piecewise-linear control, which is plotted in Figure~\ref{velocity001}. In comparison, the objective values corresponding to the ``immediate closure'' strategy (in which the valve is closed abruptly) and the ``constant closure rate'' strategy (in which the valve is closed steadily at a constant rate) are $1.4069 \times 10^{15}$ and $7.5321\times 10^4$, respectively---both much higher than the objective values in Table~\ref{tab002}.

\begin{figure}
\centering\includegraphics[scale=0.75]{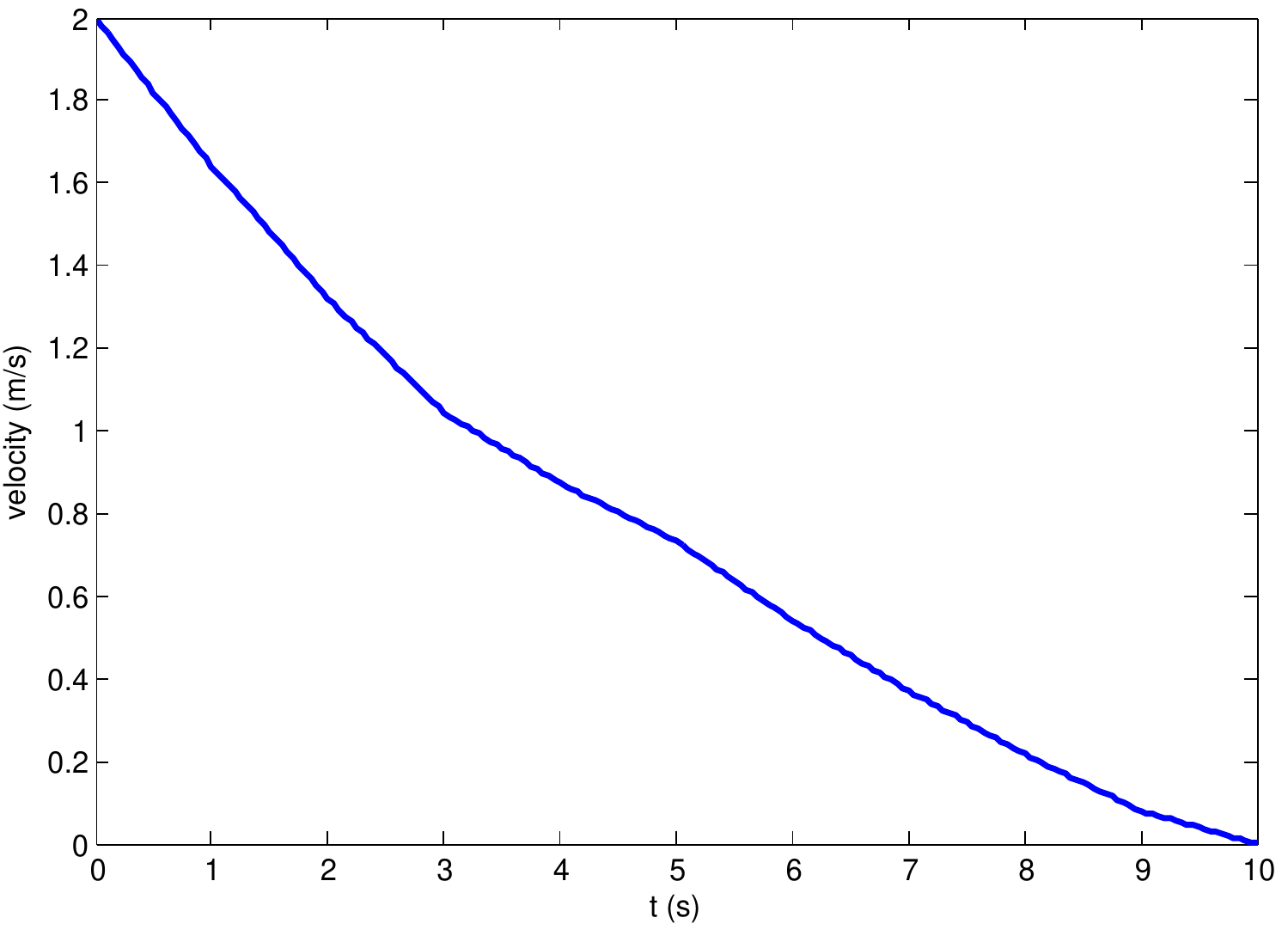}
\caption{Optimal piecewise-linear control for $N=24$}
\label{velocity001}
\end{figure}

\begin{table}
  \centering
  \caption{Optimal objective values  for Section \ref{Piecewise-Linear_1} (piecewise-linear control parameterization)}
  \label{tab002}
  \begin{tabular}{c c c c c c}
   \cmidrule{1-6}
    $N$   & 16    & 18      & 20 & 22    & 24         \\ 
    \cmidrule{1-6}
     $J(\sigma)$  & $64831$ & $60913$  & $37934$      & $29243$ & $25201$         \\
        \cmidrule{1-6}
  \end{tabular}
\end{table}

\begin{table}
  \centering
  \caption{ Optimal control parameters for Section \ref{Piecewise-Linear_1} with $N=24$ (piecewise-linear control parameterization)}
  \label{tab2}
  \begin{tabular}{c c c c c c}
   \cmidrule{1-6}
    $k$   & 1    & 2      & 3 & 4    & 5         \\ 
    \cmidrule{1-6}
     ${{\sigma}} ^k$  & $-0.3582$   & $-0.3214$      & $-0.2771$ & $-0.1700$   & $-0.1405 $      \\
     \cmidrule{1-6}
      $ k$    &6 & 7    & 8      & 9 & 10    \\
       \cmidrule{1-6}
    ${{\sigma}} ^k$    & $-0.1923$ & $-0.1702$  & $-0.1533$    & $-0.3750$      & $-0.0795$   \\
   \cmidrule{1-6}
  \end{tabular}
\end{table}

\subsection{Piecewise-Quadratic Control Parameterization}\label{Piecewise-Quadratic}

We set $\alpha=10^{-6}$ as the smoothing parameter and $\omega=1$ as the penalty parameter. We observed that ODE23 in MATLAB performs poorly in the piecewise-quadratic case. Thus, we changed the code to use ODE15s instead of ODE23 to solve the state and sensitivity systems. To determine good initial values for $\hat{\sigma}^k$, we constructed an initial piecewise-quadratic function (with smooth derivative) to approximate the optimal piecewise-linear control. This piecewise-quadratic function interpolates the optimal piecewise-linear control at the temporal knot points, and their derivatives are equal at the initial time. After constructing the initial piecewise-quadratic control, the corresponding initial values of $\hat{\sigma}^k$ were subsequently obtained. The optimal objective function values for $r=10$ and $N=16,18,20,22,24$ are given in Table~\ref{table003}. The optimal control parameters for $N=24$ are given in Table~\ref{tab4} and the corresponding optimal piecewise-quadratic control is shown in Figure~\ref{compulsivecontrol}. The pressure profiles at the pipeline terminus for the optimal piecewise-quadratic control, the optimal piecewise-linear control, and the constant closure rate control strategy are compared in Figure~\ref{comparison-pressure001}. It is clearly apparent from the figure that the piecewise-quadratic strategy results in the smoothest pressure profile with the least fluctuation. Figure~\ref{control-terminal} gives another comparison between the different control strategies for the pressure profile along the pipeline at the terminal time $t=10$s. Moreover, Figures~\ref{noncontrol1}-\ref{compulsivecontrol1} show the evolution of the pressure profile over the time and space domains, for each of the four control strategies: immediate closure, constant closure rate, optimal piecewise-linear, and optimal piecewise-quadratic. As expected, the pressure profile for the immediate closure strategy is the most volatile.

\begin{table}
  \centering
  \caption{Optimal objective values for Section \ref{Piecewise-Quadratic}  (piecewise-quadratic control parameterization)}
  \label{table003}
  \begin{tabular}{c c c c c c}
   \cmidrule{1-6}
    $N$   &16  & 18    & 20      & 22 & 24            \\ 
    \cmidrule{1-6}
     $J(\hat{\sigma})$  & $15262$   & $13911 $      & $10192 $ & $10190$    & $10187 $      \\
   \cmidrule{1-6}
  \end{tabular}
\end{table}
\begin{table}
  \centering
  \caption{Optimal control parameters for Section~\ref{Piecewise-Quadratic} with $N=24$ (piecewise-quadratic control parameterization)}
  \label{tab4}
  \begin{tabular}{c c c c c c}
   \cmidrule{1-6}
    $k$   & 1    & 2      & 3 & 4    & 5         \\ 
    \cmidrule{1-6}
     ${{\hat{\sigma}}} ^k$  & $-0.0577$   & $-0.1220$      & $-0.0096$ & $0.0551$   & $-0.0614 $      \\
     \cmidrule{1-6}
      $ k$    &6 & 7    & 8      & 9 & 10    \\
       \cmidrule{1-6}
    ${{\hat{\sigma}}} ^k$    & $-0.0172$ & $-0.0191$  & $-0.0032$    & $-0.0324$      & $0.0943$   \\
   \cmidrule{1-6}
  \end{tabular}
\end{table}
\begin{figure}
\centering\includegraphics[scale=0.75]{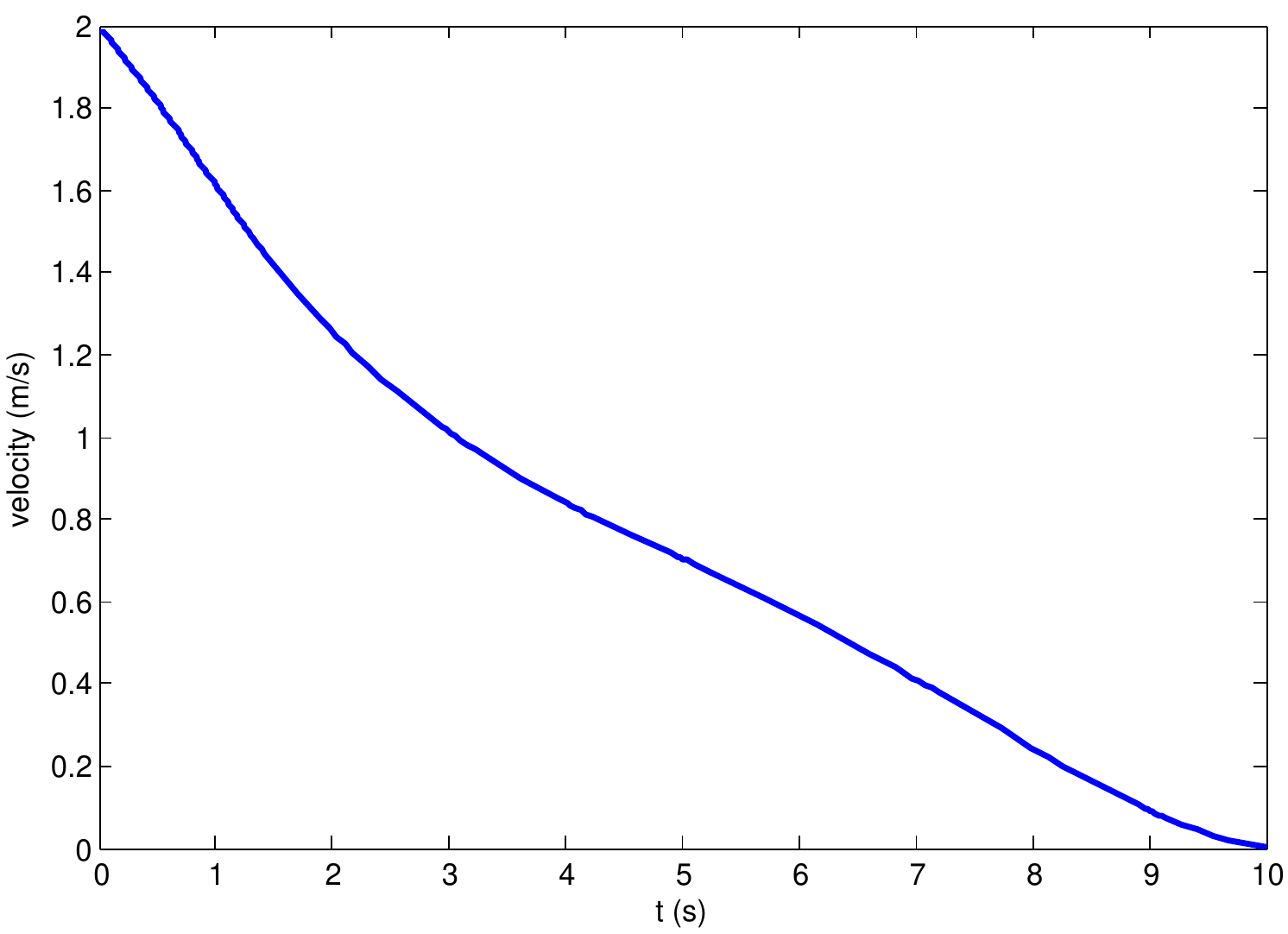}
\caption{Optimal piecewise-quadratic control for $N=24$ }
\label{compulsivecontrol}
\end{figure}
\begin{figure}
\centering\includegraphics[scale=0.75]{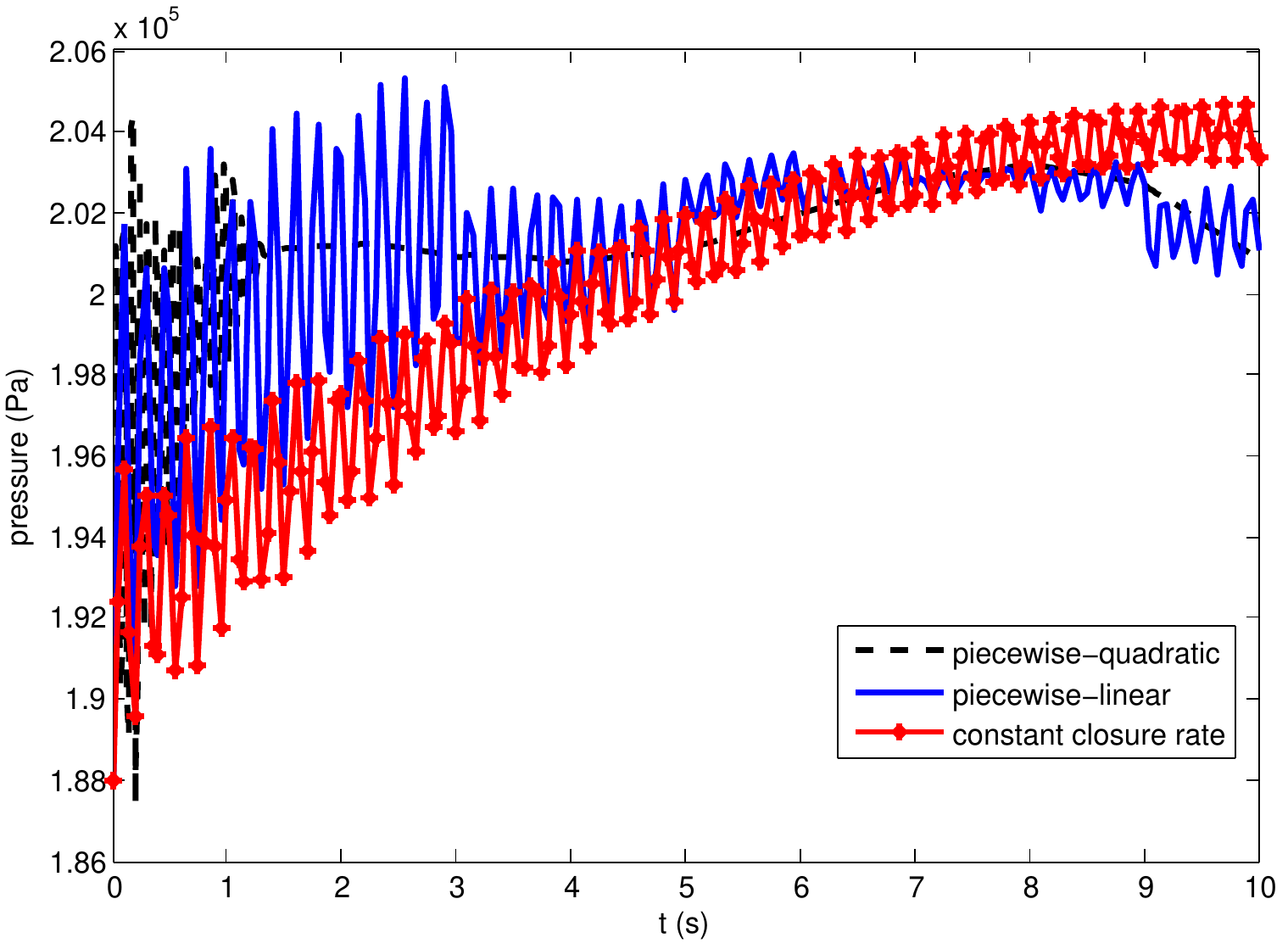}
\caption{Pressure at the pipeline terminus corresponding to the optimal piecewise-quadratic strategy, the optimal piecewise-linear strategy and the constant closure rate strategy (all for $N=24$)}
\label{comparison-pressure001}
\end{figure}


\begin{figure}
\centering\includegraphics[scale=0.75]{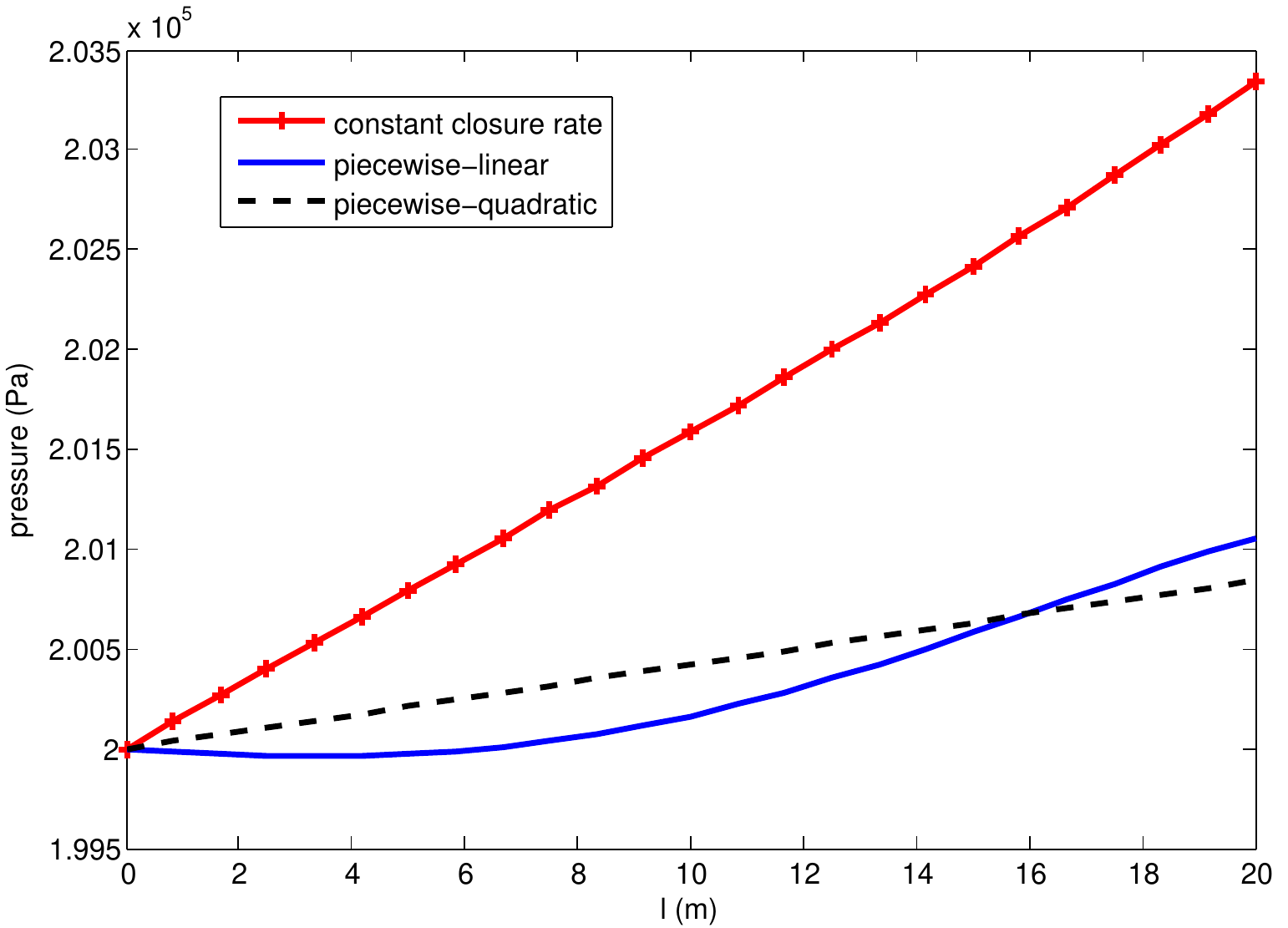}
\caption{Pressure along the pipeline at the terminal time $t=10$s corresponding to the optimal piecewise-quadratic strategy, the optimal piecewise-linear strategy and the constant closure rate strategy (all for $N=24$)}
 \label{control-terminal}
\end{figure}
\begin{figure}
\centering\includegraphics[scale=0.8]{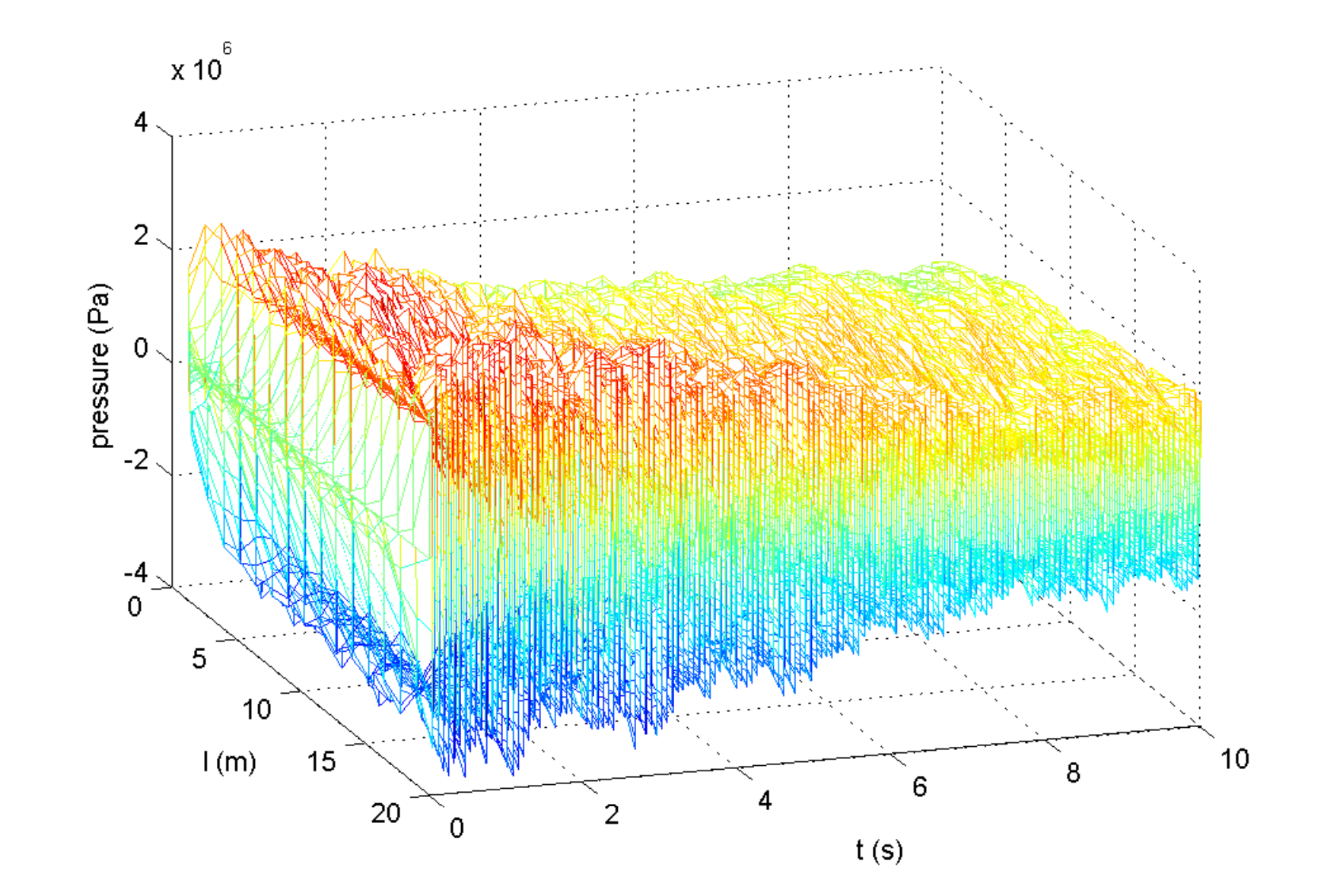}
\caption{Pressure profile corresponding to the immediate closure strategy}
\label{noncontrol1}
\end{figure}
\begin{figure}
\centering\includegraphics[scale=0.8]{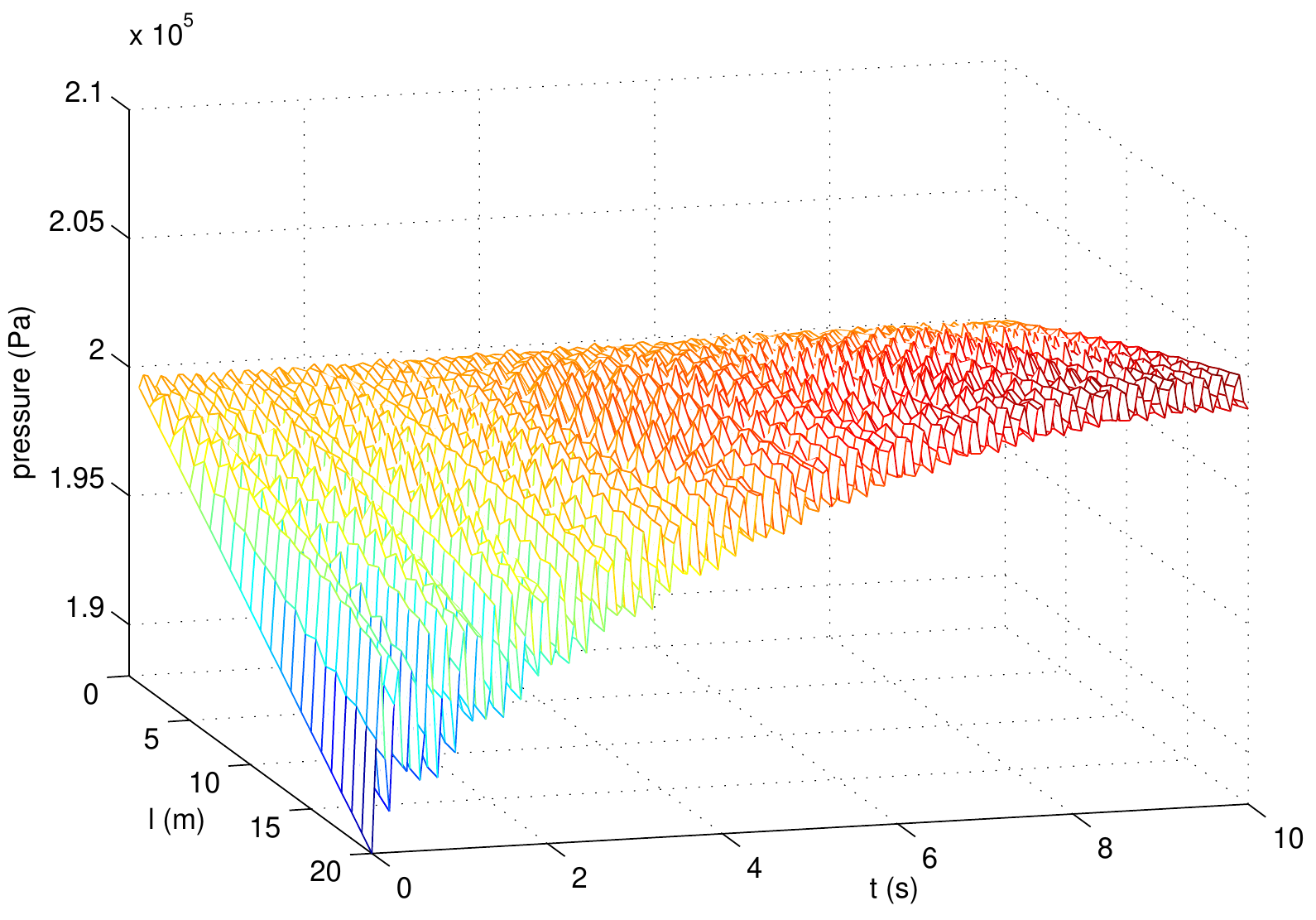}
\caption{Pressure profile corresponding to the constant closure rate strategy}
\label{linearclosure}
\end{figure}
\begin{figure}
\centering\includegraphics[scale=0.8]{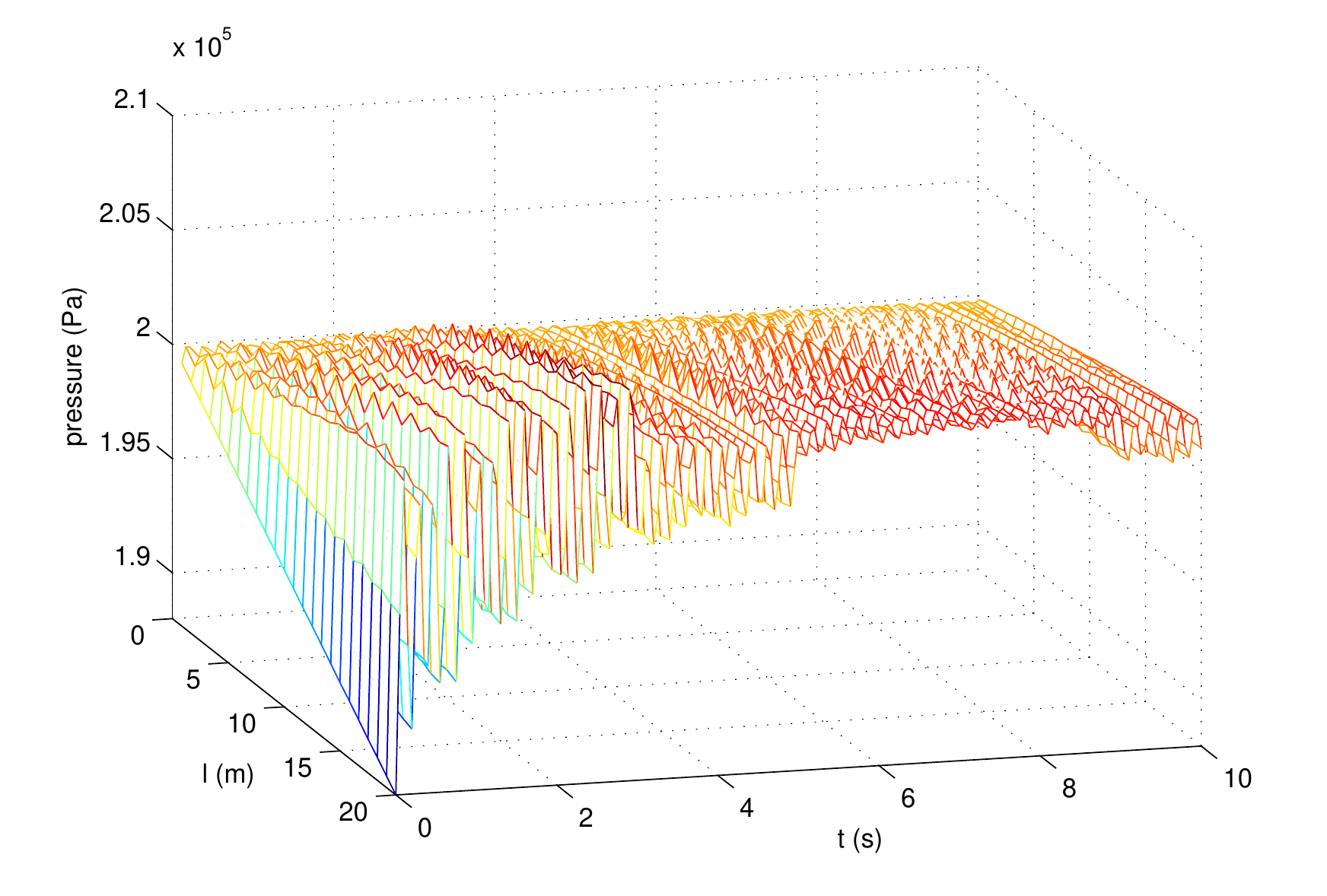}
\caption{Pressure profile corresponding to the optimal piecewise-linear strategy}
\label{transactioncontrol}
\end{figure}
\begin{figure}
\centering\includegraphics[scale=0.8]{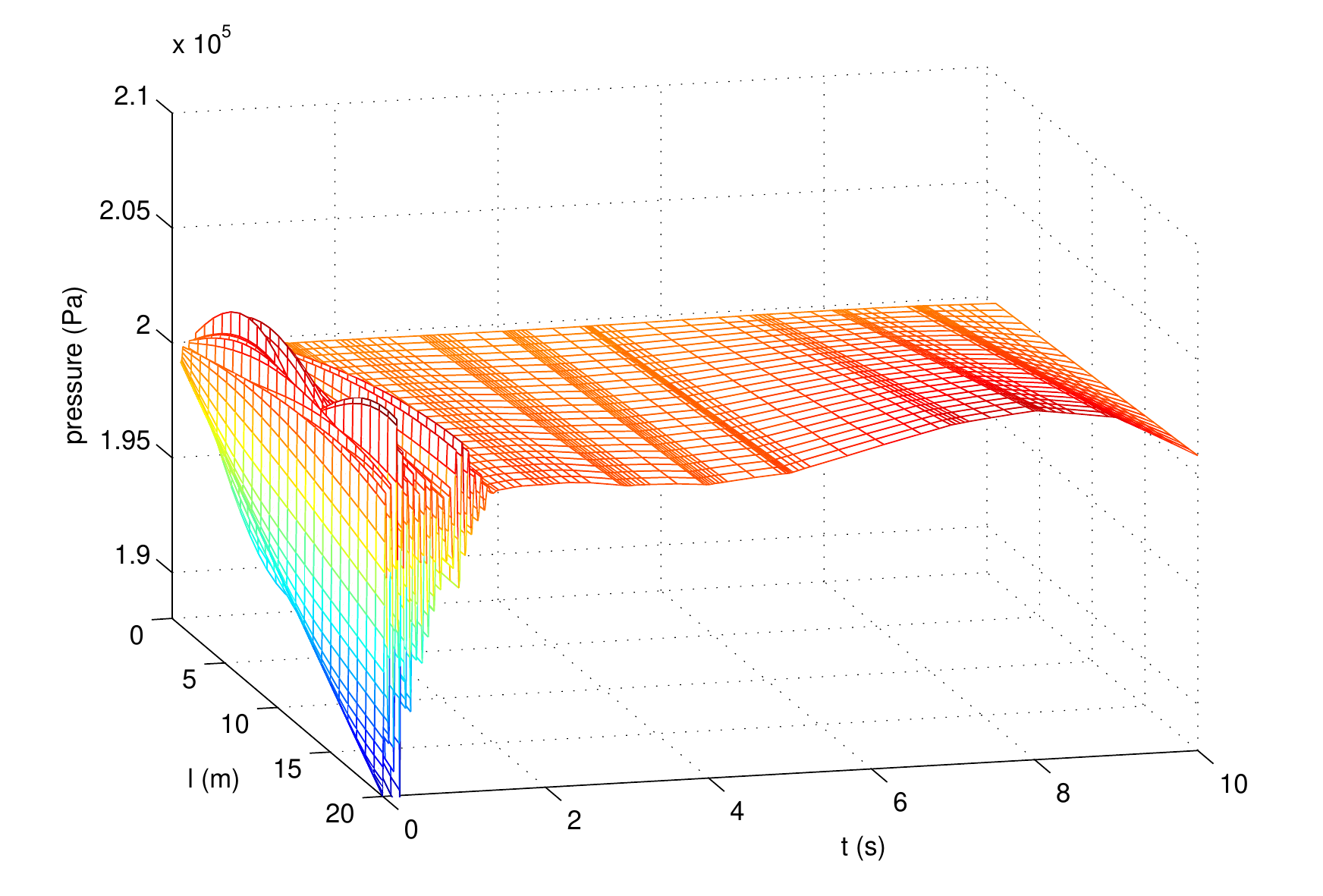}
\caption{Pressure profile corresponding to the optimal piecewise-quadratic strategy}
\label{compulsivecontrol1}
\end{figure}

\section{Conclusions and Future Work}
This paper has presented an effective computational method for solving a finite-time optimal control problem for water hammer suppression during valve closure in fluid pipelines.
The method is based on a combination of the method of lines (for discretizing the fluid flow PDEs)  and the control parameterization method (for discretizing the boundary control function). By applying these two methods in conjunction, the  optimal control
problem is reduced to an optimal parameter selection problem that can be solved using numerical optimization algorithms.
Simulation results demonstrate that this approach is highly effective at mitigating water hammer. Note that our proposed approach involves first discretizing the PDEs to obtain a set of ODEs, and then applying ODE optimal control techniques to determine the optimal valve actuation strategy. An alternative approach would be to apply PDE optimal control techniques directly to the fluid flow model. This will be considered in future work. We also plan to investigate various modifications of the objective function (\ref{obj}). For example, the different terms in  (\ref{obj}) can be  assigned different weights, or the objective can be changed to  track a given velocity profile rather than a pressure profile.
\section*{References}
\bibliographystyle{unsrt}
\bibliography{PlamaControlProblem}
\end{CJK*}
\end{document}